\documentclass[journal]{IEEEtran}

\usepackage{graphicx}

\usepackage{url}
\usepackage{float}
\usepackage{multirow,booktabs}
\usepackage{color}
\usepackage[dvipsnames]{xcolor}
\usepackage{cite}
\usepackage{algorithm}
\usepackage{algpseudocode}
\algrenewcommand\algorithmicrequire{\textbf{Input:}}
\algrenewcommand\algorithmicensure{\textbf{Output:}}
\usepackage{amsmath,amssymb,amsfonts}
\usepackage{bm}
\usepackage[utf8]{inputenc}
\usepackage[english]{babel}
\usepackage{amsthm}

\usepackage{hyperref}
\hypersetup{colorlinks=true,linkcolor=blue,citecolor=blue,urlcolor=blue}

\usepackage{caption}
\usepackage{subcaption}
\theoremstyle{plain}
\newtheorem{theorem}{Theorem}
\newtheorem{lemma}{Lemma}

\theoremstyle{definition}

\theoremstyle{remark}
\newtheorem{remark}{Remark}

\DeclareMathOperator{\sign}{sign}
\DeclareMathOperator{\diag}{diag}
\DeclareMathOperator{\tr}{tr}
\DeclareMathOperator{\Real}{Re}

\DeclareMathOperator{\prox}{prox}

\begin{document}

%

\title{DOA Estimation from One-Bit Magnitude-only Measurements via Sign-Consistency Optimization}

\author{
\IEEEauthorblockN{
Xicheng Lu, 
 Wei Liu, \textit{Senior Member, IEEE} and Akram Alomainy, \textit{Senior Member, IEEE}
}
}

\maketitle


\begin{abstract}
The direction-of-arrival (DOA) estimation problem using one-bit quantized magnitude-only measurements is studied, where magnitude-only measurements offer the advantage of robustness against phase errors, thereby avoiding the need for array calibration, while the use of one-bit quantization significantly reduces hardware cost and system complexity. As their direct combination results in meaningless constant measurements, we have formulated  a sign-consistency optimization problem using a smooth logistic surrogate with $\ell_{2,1}$-norm regularization to promote joint sparsity. To solve this, a proximal-gradient algorithm is developed with guaranteed convergence to a critical point. Numerical results demonstrate that the proposed method achieves accuracy comparable to coherent one-bit baselines under ideal conditions, while maintaining robust performance under severe phase errors that substantially impair coherent methods.
\end{abstract}

\let\thefootnote\relax\footnote{This work has been submitted to the IEEE for possible publication. Copyright may be transferred without notice, after which this version may no longer be accessible.

X. Lu and A. Alomainy are with the School of Electronic Engineering and Computer Science, Queen Mary University of London, London, UK (e-mail:xicheng.lu@qmul.ac.uk, a.alomainy@qmul.ac.uk).

W. Liu is with the Department of Electrical and Electronic Engineering, The Hong Kong Polytechnic University, Hong Kong SAR (email: wei2.liu@polyu.edu.hk).}
%


%

\section{Introduction}

Direction-of-arrival (DOA) estimation is a fundamental problem in radar, sonar,
and wireless communications~\cite{vanTrees2002optimum,krim1996parametric,liu2023review}.
Traditional high-resolution algorithms such as subspace-based methods~\cite{schmidt1986music,roy1989esprit},
and sparsity-based formulations~\cite{malioutov2005sparse,stoica2010spice,kim2012compressive,zheng2013sparse,shen2016underdetermined}, assume access to complex-valued snapshots with phase information preserved and a well-calibrated array.

In practice, sensor gain mismatches and phase errors will inevitably degrade their performance, and many robust DOA estimation methods have been developed in this direction, such as those based on eigenstructure~\cite{weiss1990eigenstructure,liu2011eigenstructure}, sparse Bayesian learning (SBL)~\cite{zhao2013improved,liu2013unified}, and atomic norm minimization~\cite{chen2020new}. However, they still rely on coherent complex measurements and attempt to mitigate phase errors through enhanced modeling rather than eliminating the dependence on phase information.

A more recent approach is to estimate DOAs directly from magnitude/intensity-only measurements, where phase information is inherently unavailable. This non-coherent setting is closely related to the phase retrieval problem, where one aims to recover a signal from the magnitudes of linear measurements~\cite{candes2015phase,shechtman2015phase}. Early methods considered sparse nonlinear least-squares formulations solved via local search or greedy procedures~\cite{kim2014noncoherent}, as well as harmonic spectral estimation from magnitude-squared data~\cite{zayyani2016noncoherent}.

However, the lack of phase information may introduce intrinsic ambiguities in magnitude-only DOA estimation, including the mirroring, spatial shift, and spatial order ambiguities~\cite{kim2014noncoherent,wan2021dual}. To remove these ambiguities, one or more reference sources with known DOA can be employed~\cite{kim2014noncoherent,zayyani2016noncoherent}. Other approaches square the sample covariance matrix over multiple snapshots to obtain a linearized sparse recovery model~\cite{tian2019noncoherent}, or incorporate group-sparsity regularization within majorization-minimization to enhance robustness~\cite{wan2021dual}. Alternatively, dual uniform linear array (ULA) geometries can resolve all three ambiguities via inter-subarray magnitude discrepancies~\cite{wan2021dual}, while uniform circular arrays (UCAs) can avoid the mirroring ambiguity due to the lack of conjugate symmetry in their steering vectors~\cite{wan2021uca,wan2023target}. Despite these advances, existing magnitude-only DOA methods are normally developed under the assumption of full-precision measurements.

In parallel, low-resolution analog-to-digital converters (ADCs), including one-bit, 1.5-bit, and two-bit quantizers, have attracted growing interest due to their reduced power consumption, hardware complexity, and data storage requirements~\cite{lu2025tsp15bit}. Among these, one-bit quantization has been most extensively studied. In subspace-based methods, Bussgang's theorem~\cite{bussgang1952crosscorr} and the Van~Vleck arcsine law~\cite{vanvleck1966clipped} establish a closed-form relationship between the correlation of sign-quantized Gaussian data and that of the unquantized one. Built on this, the arcsine law was applied to sparse arrays for one-bit DOA estimation in~\cite{liu2017one}, and one-bit MUSIC was proposed by exploiting this relationship to enable subspace-based DOA estimation from one-bit measurements~\cite{huang2019one}. Subsequent work generalized covariance recovery to non-zero thresholds~\cite{liu2021onebit}, estimators for time-varying thresholds under both non-stationary and stationary signal models were proposed in~\cite{eamaz2022nonstationary,eamaz2023stationary}, and a unified algorithm with analytical performance guarantees was presented in~\cite{xiao2023onebit}. From a compressed-sensing perspective, one-bit DOA estimation was formulated as a sparse recovery problem with sign-consistency constraints in~\cite{stockle2015spawc}, where the complex-valued binary iterative hard thresholding (CBIHT-$\ell_2$) algorithm was proposed. Joint sparse representation was exploited for improved signal reconstruction in~\cite{yu2016splet}. More recently, the Robust One-bit Compressed Sensing (ROCS) framework~\cite{li2024rocs} employed a $\tanh$ surrogate to approximate the discontinuous sign function, enabling robust gradient-based recovery under model mismatches.

On the other hand, one-bit phase retrieval has received considerable attention from the mathematical perspective. Early work established that sparse signals can be recovered from one-bit magnitude measurements using greedy refinement procedures~\cite{mroueh2013quantization}. Subsequent algorithmic developments include approximate message passing tailored to quantized phase retrieval~\cite{zhu2019phase} and variants of Wirtinger Flow adapted for binary observations~\cite{kishore2020wirtinger}. The interplay between sample complexity and computational cost was analyzed in~\cite{eamaz2022one}, showing that more measurements can reduce the iterations required for convergence. More recently, minimax-optimal recovery rates and computationally efficient algorithms have been established in~\cite{chen2024one}.

Despite this progress, a clear gap remains between these two research directions. Existing magnitude-only DOA methods~\cite{kim2014noncoherent,wan2021dual,wan2021uca} are based on full-precision measurements and cannot exploit the hardware benefits of low-resolution ADCs, while prevalent one-bit DOA estimation schemes~\cite{stockle2015spawc,huang2019one,li2024rocs} rely on coherent processing that preserves phase information and thus degrade severely under sensor phase errors. In particular, coherent one-bit methods suffer from the same performance degradation under phase errors as full-precision methods. Conversely, full-precision magnitude-only methods discard phase information entirely for robustness, but do not address the practical need for low-cost, low-power consumption. To the best of our knowledge, no existing work jointly addresses one-bit quantization and magnitude-only processing for DOA estimation. Motivated by the robustness of magnitude-only measurements to phase errors and the hardware efficiency of one-bit quantization, we study one-bit magnitude-only DOA estimation in this work, bridging the gap between hardware-efficient sampling and phase-error-robust non-coherent processing. 

The methods on one-bit phase retrieval introduced earlier above address a closely related problem of recovering signals from one-bit quantized magnitudes of linear projections, and the non-zero threshold requirement identified therein (since $\mathrm{sign}(|z|) = 1$ for all $z \neq 0$ under zero thresholds) applies equally to our present DOA estimation problem;  however, these methods assume random measurement matrices, and their theoretical guarantees do not directly extend to the deterministic structured steering matrix encountered in array signal processing, nor do they exploit the property of multi-snapshot joint sparsity that is central to DOA estimation.  Based on this insight, we cast the studied problem as a structured instance of one-bit phase retrieval that leverages array geometry and joint sparsity across snapshots, and develop a sign-consistency formulation for one-bit magnitude-only DOA estimation with UCAs. A smoothed proximal-gradient algorithm, termed \textbf{O}ne-\textbf{BI}t \textbf{M}agnitude-\textbf{O}nly \textbf{D}OA \textbf{E}stimation via \textbf{S}ign-consistency and group sparsi\textbf{T}y (\textbf{OBI-MODEST}), is proposed to solve the resulting nonconvex group-sparse problem. Under ideal conditions without sensor phase errors, OBI-MODEST achieves a DOA estimation accuracy comparable to representative coherent one-bit sparse recovery solutions (e.g., CBIHT-$\ell_2$~\cite{stockle2015spawc} and ROCS~\cite{li2024rocs}), despite using only one-bit magnitude information. More importantly, under sensor phase errors, the proposed scheme and algorithm remain robust, offering a favorable trade-off between reliability and hardware complexity.

The main contributions are summarized as follows.
\begin{itemize}
    \item \textbf{Problem formulation and motivation:}
    We formulate DOA estimation from one-bit quantized magnitude-only measurements with non-zero thresholds. Since taking magnitudes removes phase information, this signal model is inherently immune to per-sensor phase errors. We first show that conventional magnitude-fitting objectives are ill-suited to one-bit data because only binary threshold comparisons are available, and then establish a local constancy property showing that small perturbations in the signal matrix do not change the one-bit output, which causes zero gradients and motivates a continuous surrogate.
    
    \item \textbf{Proposed algorithm:}
    We propose a sign-consistency formulation that combines a logistic surrogate loss~\cite{plan2013one,boufounos2008onebit} with a smoothed magnitude operator to ensure differentiability, together with $\ell_{2,1}$-norm regularization to exploit joint sparsity across snapshots. We develop the OBI-MODEST algorithm based on proximal-gradient iterations and derive a closed-form shrinkage update that enables efficient computation.
    
    \item \textbf{Convergence analysis and simulation validation:}
    We establish Lipschitz continuity of the gradient and prove convergence to a critical point. Simulations on uniform circular arrays show that OBI-MODEST achieves accuracy comparable to coherent one-bit baselines under ideal conditions. Under sensor phase errors, coherent methods degrade by more than an order of magnitude while OBI-MODEST remains stable, demonstrating the robustness advantage of magnitude-only processing.
\end{itemize}

The remainder of this paper is organized as follows. 
The one-bit magnitude-only DOA estimation problem is formulated in Section~\ref{sec:problem_formulation}, and the developed OBI-MODEST algorithm with its convergence properties are presented in Section~\ref{sec:proposed_method}. 
Section~\ref{sec:simulation} provides simulation results, and 
Section~\ref{sec:conclusion} concludes the paper. A list of notations is provided in Table~\ref{tab:notation}.

\begin{table}[t]
    \centering
    \caption{List of Notations}
    \label{tab:notation}
    \begin{tabular}{l|l}
        \toprule
        \textbf{Symbol} & \textbf{Description} \\
        \midrule
        $\mathbf{A},\,\mathbf{a},\,a$ & Matrices, column vectors, and scalars \\
        $[\mathbf{A}]_{i,:}$ & $i$-th row of matrix $\mathbf{A}$ \\
        $[\mathbf{A}]_{:,j}$ & $j$-th column of matrix $\mathbf{A}$ \\
        $[\mathbf{A}]_{i,j}$ & $(i,j)$-th entry of matrix $\mathbf{A}$ \\
        $[\mathbf{a}]_{i}$ & $i$-th entry of vector $\mathbf{a}$ \\
        $(\cdot)^{\mathsf T}$ & Transpose of a real matrix/vector \\
        $(\cdot)^{\mathsf H}$ & Hermitian (conjugate) transpose \\
        $(\cdot)^\ast$ & Complex conjugate \\
        $\mathbb{R}^{M\times N}$ & Real-valued $M\times N$ matrices \\
        $\mathbb{C}^{M\times N}$ & Complex-valued $M\times N$ matrices \\
        $\mathbf{0},\,\mathbf{I}_n$ & All-zero matrix; $n\times n$ identity matrix \\
        $\mathbf{1}_n$ & All-one column vector of length $n$ \\
        $|z|$ & Magnitude of a complex scalar $z$ \\
        $|\mathbf{Z}|$ &  magnitude of a complex matrix \\
        $|z|_\epsilon$ & Smoothed magnitude: $\sqrt{|z|^2+\epsilon^2}$ \\
        $\odot$ & Hadamard product \\
        $\langle\mathbf{A},\mathbf{B}\rangle$ & Frobenius inner product: $\Real\,\tr(\mathbf{A}^{\mathsf H}\mathbf{B})$ \\
        $\|\cdot\|_2$ & Euclidean norm \\
        $\|\cdot\|_F$ & Frobenius norm of a matrix \\
        $\|\cdot\|_\infty$ &  $\|\mathbf{Z}\|_\infty \!=\! \max_{i,j}|[\mathbf{Z}]_{i,j}|$ \\
        $\|\cdot\|_{2,1}$ & Mixed $\ell_{2,1}$ norm: $\sum_g\|[\mathbf{X}]_{g,:}\|_2$ \\
        $\mathbb{I}(\cdot)$ & Indicator function ($1$ if true, $0$ otherwise) \\
        $\sign(\cdot)$ &  sign ($+1$ if $\ge0$, $-1$ otherwise) \\
        $\nabla f(\cdot)$ & Gradient of function $f$ \\
        $\partial f(\cdot)$ & Limiting subdifferential of $f$ \\
        $\prox_{\gamma f}(\cdot)$ & Proximal operator of $f$ with parameter $\gamma>0$ \\
        $\mathcal{O}(\cdot)$ & Big-O (asymptotic order) notation \\
        \bottomrule
    \end{tabular}
\end{table}

\section{Problem Formulation}
\label{sec:problem_formulation}

\subsection{Signal Model}
\label{subsec:signal_model}

Consider $K$ narrowband far-field sources with wavelength $\lambda$
impinging on an M-sensor UCA with a radius $R$. The $m$-th sensor is placed at the angular
position $\gamma_m = 2\pi(m-1)/M$, $m=1,\ldots,M$. The wavenumber is
$k_0 = 2\pi/\lambda$, and with $\xi = k_0R$, the array steering vector for
an azimuth angle $\theta$ is
\begin{equation}
  \mathbf{a}(\theta)
  =
  \big[e^{\mathrm{j}\xi\cos(\theta-\gamma_1)},\ldots,
      e^{\mathrm{j}\xi\cos(\theta-\gamma_M)}\big]^{\mathsf T}
  \in\mathbb{C}^{M}.
\label{eq:uca_steer}
\end{equation}

Let $\boldsymbol{\theta}=\{\theta_1,\ldots,\theta_K\}$ denote the set of
true DOAs in the angular search domain $\mathcal{D}$, and define the
steering matrix
\begin{equation}
  \mathbf{A}(\boldsymbol{\theta})
  =
  [\mathbf{a}(\theta_1),\ldots,\mathbf{a}(\theta_K)]
  \in\mathbb{C}^{M\times K}.
\end{equation}
At snapshot index $p\in\{1,\ldots,P\}$, the complex amplitude of the
$k$-th source is denoted by $s_k[p]\in\mathbb{C}$, and the source signal
vector is
$\mathbf{s}[p] = [\,s_1[p],\ldots,s_K[p]\,]^{\mathsf T} \in\mathbb{C}^{K}$.

Sensor phase errors are modeled by a diagonal matrix $\mathbf{D} = \operatorname{diag}(e^{\mathrm{j}\phi_{1}}, \ldots, e^{\mathrm{j}\phi_{M}})$, where $\phi_{m}$ denotes the phase error of the $m$-th sensor. The received signal vector (without noise) is thus
\begin{equation}
  \mathbf{x}[p] = \mathbf{D}\mathbf{A}(\boldsymbol{\theta})\,\mathbf{s}[p] \in\mathbb{C}^{M}.
\end{equation}
A key property of magnitude-only processing is that the phase error matrix $\mathbf{D}$ does not affect the magnitude measurements. Since $\mathbf{D}$ is diagonal, the $m$-th entry of $\mathbf{D}\mathbf{x}$ equals $e^{\mathrm{j}\phi_m} x_m$, whose magnitude satisfies $|e^{\mathrm{j}\phi_m} x_m| = |e^{\mathrm{j}\phi_m}| \cdot |x_m| = |x_m|$. Consequently, the  magnitude is invariant to phase errors:
\begin{equation}
  |\mathbf{D}\mathbf{A}(\boldsymbol{\theta})\,\mathbf{s}[p]| = |\mathbf{A}(\boldsymbol{\theta})\,\mathbf{s}[p]|.
\label{eq:phase_invariance}
\end{equation}
The corresponding magnitude-only measurement vector at snapshot $p$ is modeled as
\begin{equation}
  \mathbf{y}[p]
  =
  |\mathbf{x}[p]|
  +
  \mathbf{n}[p]
  =
  \big|\mathbf{A}(\boldsymbol{\theta})\,\mathbf{s}[p]\big|
  +
  \mathbf{n}[p],
\label{eq:mag_cont}
\end{equation}
where $\mathbf{n}[p]\in\mathbb{R}^{M}$ denotes zero-mean additive white Gaussian noise with variance $\sigma_n^2$. 

To enable sparse representation, we discretize
$\mathcal{D}$ into an angular grid
$\Theta=\{\vartheta_g\}_{g=1}^{G}$ with $G\gg K$, and form the
overcomplete steering matrix
\begin{equation}
  \mathbf{A}(\Theta)
  =
  [\mathbf{a}(\vartheta_1),\ldots,\mathbf{a}(\vartheta_G)]
  \in\mathbb{C}^{M\times G}.
\end{equation}
The source signal vector $\mathbf{s}[p]$ admits a sparse representation
on the angular grid $\Theta$ as
\begin{equation}
  \tilde{\mathbf{s}}[p]
  =
  [\,\tilde{s}_1[p],\ldots,\tilde{s}_G[p]\,]^{\mathsf T}
  \in\mathbb{C}^{G},
\end{equation}
where $\tilde{s}_g[p]$ denotes the coefficient associated with grid angle
$\vartheta_g$ at snapshot $p$. Under this representation, the magnitude-only model in~\eqref{eq:mag_cont} becomes
\begin{equation}
  \mathbf{y}[p]
  =
  \big|\mathbf{A}(\Theta)\,\tilde{\mathbf{s}}[p]\big|
  +
  \mathbf{n}[p].
\label{eq:mag_grid_vec}
\end{equation}
The source signals are assumed to be uncorrelated with power $\sigma_s^2$.
Collecting $P$ snapshots yields the multiple-measurement-vector (MMV) model
\begin{equation}
  \mathbf{Y}
  =
  \big|\mathbf{A}(\Theta)\,\tilde{\mathbf{S}}\big|
  +
  \mathbf{N},
\label{eq:R_model}
\end{equation}
where
$\mathbf{Y}
  =
  [\mathbf{y}[1],\ldots,\mathbf{y}[P]]
  \in\mathbb{R}^{M\times P}$
denotes the magnitude-only data matrix,
$\tilde{\mathbf{S}}
  =
  [\tilde{\mathbf{s}}[1],\ldots,\tilde{\mathbf{s}}[P]]
  \in\mathbb{C}^{G\times P}$
is the sparse signal matrix, and
$\mathbf{N}
  =
  [\mathbf{n}[1],\ldots,\mathbf{n}[P]]
  \in\mathbb{R}^{M\times P}$ is the noise matrix. Since the $K$ sources share the same DOAs across all snapshots, the matrix $\tilde{\mathbf{S}}$ exhibits row sparsity that only $K \ll G$ rows are nonzero, with the indices of these active rows corresponding to the grid angles closest to the true DOAs. 

Each magnitude-only snapshot is then quantized by a 
one-bit comparator:
\begin{equation}
  \mathbf{y}_{\mathrm{1bit}}[p]
  =
  \operatorname{sign}\!\big(\mathbf{y}[p]-\boldsymbol{\tau}\big),
\label{eq:quant_vec}
\end{equation}
where
$\boldsymbol{\tau}=[\tau_1,\ldots,\tau_M]^{\mathsf T}\in\mathbb{R}^{M}$
is the vector of quantization thresholds, assumed known, and $\operatorname{sign}(t)=+1$ for $t\ge 0$ and $-1$
otherwise. Collecting the one-bit magnitude-only measurements as $\mathbf{Y}_{\mathrm{1bit}}
  =
  [\mathbf{y}_{\mathrm{1bit}}[1],\ldots,\mathbf{y}_{\mathrm{1bit}}[P]]
  $
and with $\mathbf{T}=\boldsymbol{\tau}\,\mathbf{1}_P^{\mathsf T}$, the one-bit quantization model is
\begin{equation}
  \mathbf{Y}_{\mathrm{1bit}}
  =
  \operatorname{sign}\big(\mathbf{Y}-\mathbf{T}\big)
  =
  \operatorname{sign}\big(|\mathbf{A}(\Theta)\,\tilde{\mathbf{S}}|
                          +\mathbf{N}-\mathbf{T}\big).
\label{eq:Y_model}
\end{equation}

\subsection{Identifiability and Ambiguities}
\label{subsec:ambiguity}

Magnitude-only DOA estimation using ULAs is fundamentally limited by inherent ambiguities, most notably \emph{mirroring}, \emph{spatial-shift}, and \emph{spatial-order} ambiguities~\cite{wan2021dual,wan2021uca}. While the mirroring ambiguity in ULAs arises from the conjugate symmetry $\mathbf{a}(-\theta)=\mathbf{a}(\theta)^*$, the UCA manifold naturally avoids this issue. 

Furthermore, the spatial-shift and spatial-order ambiguities in ULAs are direct consequences of the linear phase progression inherent to the Vandermonde structure, given by $[\mathbf{a}(\theta)]_m = e^{j(m-1)\varpi}$, where $\varpi = k_0 d \sin\theta$ represents the spatial frequency. Specifically, the \emph{spatial-shift ambiguity} arises because an arbitrary linear phase shift $\delta(m-1)$ implies $[\mathbf{a}(\theta)]_m e^{j(m-1)\delta} = e^{j(m-1)(\varpi+\delta)}$, which is mathematically indistinguishable from a signal arriving from a new direction corresponding to $\varpi+\delta$. Similarly, the \emph{spatial-order ambiguity} (at $d=\lambda/2$) stems from the $2\pi$-periodicity condition $e^{j(m-1)\varpi} = e^{j(m-1)(\varpi - 2\pi)}$. If a shift causes the spatial frequency to exceed the principal interval (i.e., $\varpi + \delta > \pi$), it wraps to an equivalent frequency $\varpi' = \varpi + \delta - 2\pi$. Since $\varpi'$ corresponds to a valid angle $\theta'$ (where $\sin\theta' = \sin\theta + \delta/k_0 d - 2\pi/k_0 d$), this wrapping mechanism can alter the spatial ordering of sources without affecting the magnitude measurements. Conversely, for the UCA geometry, the adjacent phase increment varies with the sensor index $m$, as given by:
\begin{align}
\Delta\phi_m(\theta)
&\triangleq \xi\big[\cos(\theta-\gamma_{m+1})-\cos(\theta-\gamma_m)\big] \nonumber\\
&= 2\xi\sin\left(\frac{\pi}{M}\right)
\sin\left(\theta-\frac{\pi(2m-1)}{M}\right).
\label{eq:delta_phi_uca}
\end{align}
As a result, the shift-invariance required for both spatial-shift and spatial-order ambiguities is strictly broken. 

Nevertheless, UCAs exhibit a specific $180^\circ$ rotational ambiguity, as $\mathbf{a}(\theta+\pi)=\mathbf{a}(\theta)^*$ implies that $\{(\theta_k,s_k)\}$ and $\{(\theta_k+\pi,s_k^*)\}$ yield identical magnitude measurements. This is typically resolved by restricting the angular sector (e.g., to $[-\pi/2,\pi/2]$) or utilizing a reference source.

\subsection{Limitations of Direct Loss Design}
\label{subsec:why_not_amp}

We adopt the sparsity signal model in \eqref{eq:R_model}--\eqref{eq:Y_model}, where the unknown 
signal matrix $\tilde{\mathbf{S}}\in\mathbb{C}^{G\times P}$ is row-sparse with $K\ll G$ nonzero rows.
In the one-bit noncoherent setting, only the thresholded measurements $\mathbf{Y}_{\mathrm{1bit}}$ in 
\eqref{eq:Y_model} are observed. 
A natural approach is to directly match the predicted one-bit outputs 
$\widehat{\mathbf{Y}}_{\mathrm{1bit}}(\tilde{\mathbf{S}}) \triangleq \operatorname{sign}\bigl(|\mathbf{A}(\Theta)\tilde{\mathbf{S}}|-\mathbf{T}\bigr)$
to the observations:
\begin{equation}
\label{eq:ideal_01_loss}
\min_{\tilde{\mathbf{S}}}\;
\bigl\|\mathbf{Y}_{\mathrm{1bit}}-\widehat{\mathbf{Y}}_{\mathrm{1bit}}(\tilde{\mathbf{S}})\bigr\|_F^2
+\eta\|\tilde{\mathbf{S}}\|_{2,1},
\end{equation}
where $\eta>0$ is the regularization parameter. 
Since the entries of both $\mathbf{Y}_{\mathrm{1bit}}$ and $\widehat{\mathbf{Y}}_{\mathrm{1bit}}(\tilde{\mathbf{S}})$ are constrained to $\{+1, -1\}$, the squared Frobenius norm of their difference is proportional to the number of sign mismatches. 
Specifically, for any entry $(m,p)$, the squared error $([\mathbf{Y}_{\mathrm{1bit}}]_{m,p} - [\widehat{\mathbf{Y}}_{\mathrm{1bit}}(\tilde{\mathbf{S}})]_{m,p})^2$ equals $0$ if the signs match and $4$ if they differ. Consequently,
\begin{equation}
\label{eq:01_loss_equiv}
\bigl\|\mathbf{Y}_{\mathrm{1bit}}-\widehat{\mathbf{Y}}_{\mathrm{1bit}}(\tilde{\mathbf{S}})\bigr\|_F^2
 = 4 \cdot d_H\Bigl(\mathbf{Y}_{\mathrm{1bit}}, \widehat{\mathbf{Y}}_{\mathrm{1bit}}(\tilde{\mathbf{S}})\Bigr),
\end{equation}
where $d_H(\cdot, \cdot)$ denotes the Hamming distance defined as:
\begin{equation}
d_H(\mathbf{B},\mathbf{C})
\triangleq
\sum_{m=1}^{M}\sum_{p=1}^{P}
\mathbb{I}\!\left([\mathbf{B}]_{m,p}\neq[\mathbf{C}]_{m,p}\right).
\end{equation}

\begin{lemma}
\label{lem:local_constancy}
Let $\widehat{\mathbf{Y}}_{\mathrm{1bit}}(\tilde{\mathbf{S}})\triangleq
\operatorname{sign}\bigl(|\mathbf{A}(\Theta)\tilde{\mathbf{S}}|-\mathbf{T}\bigr)$ denote the one-bit prediction.
Consider any point $\tilde{\mathbf{S}}_0\in\mathbb{C}^{G\times P}$ satisfying the non-degeneracy condition (i.e., no predicted magnitude lies exactly on the threshold):
\begin{equation}
\label{eq:margin_condition}
\bigl|[\mathbf{A}(\Theta)\tilde{\mathbf{S}}_0]_{m,p}\bigr|\neq \tau_m, 
\qquad \forall\, (m,p).
\end{equation}
We define the minimum decision margin $\rho$ as
\begin{equation}
\label{eq:rho_def}
\rho \triangleq \min_{m,p} \left|
\bigl|[\mathbf{A}(\Theta)\tilde{\mathbf{S}}_0]_{m,p}\bigr|-\tau_m
\right|.
\end{equation}
Condition~\eqref{eq:margin_condition} guarantees that $\rho > 0$.
The one-bit output is locally constant, i.e., $\widehat{\mathbf{Y}}_{\mathrm{1bit}}(\tilde{\mathbf{S}})=\widehat{\mathbf{Y}}_{\mathrm{1bit}}(\tilde{\mathbf{S}}_0)$, provided that the magnitude perturbation satisfies
\begin{equation}
\label{eq:local_delta_condition}
\bigl\||\mathbf{A}(\Theta)\tilde{\mathbf{S}}|-|\mathbf{A}(\Theta)\tilde{\mathbf{S}}_0|\bigr\|_\infty
< \rho.
\end{equation}
Consequently, a sufficient condition in the variable domain is given by
\begin{equation}
\label{eq:local_ball_condition}
\|\tilde{\mathbf{S}}-\tilde{\mathbf{S}}_0\|_F
<
\frac{\rho}{\|\mathbf{A}(\Theta)\|_2}.
\end{equation}
\end{lemma}
The proof is provided in Appendix~\ref{app:proof_local_constancy}.

Equation~\eqref{eq:01_loss_equiv} shows that the squared error in
\eqref{eq:ideal_01_loss} is proportional to the Hamming distance, i.e., it
counts the number of sign mismatches. Therefore, reducing the loss essentially
requires discrete sign-flip decisions, which results in a nonsmooth and
combinatorial optimization landscape.

Moreover, Lemma~\ref{lem:local_constancy} shows that
$\operatorname{sign}(|\mathbf A(\Theta)\tilde{\mathbf S}|-\mathbf T)$ is locally
constant away from the thresholds. Consequently, the objective in
\eqref{eq:ideal_01_loss} is piecewise constant over large regions and provides
little useful descent information, so proximal-gradient iterations can be employed.

These observations motivate us replacing $\operatorname{sign}(\cdot)$ by a smooth approximation.

\section{Proposed Method}
\label{sec:proposed_method}

\subsection{Sign-Consistency Formulation and Algorithm}
\label{subsec:formulation_algorithm}
We now introduce a smooth surrogate that
retains the same threshold-side consistency while enabling gradient-based optimization.

Define the thresholded magnitude residual matrix
\begin{equation}
\label{eq:residual_matrix}
\mathbf{R}(\tilde{\mathbf{S}})
\triangleq
|\mathbf{A}(\Theta)\tilde{\mathbf{S}}|-\mathbf{T}
\in\mathbb{R}^{M\times P},
\end{equation}
where $[\mathbf{T}]_{m,p}=\tau_m$ for all $p$. A natural consistency requirement is that the predicted
magnitudes lie on the same side of the thresholds as indicated by the one-bit observations:
\begin{equation}
\label{eq:sign_consistency}
[\mathbf{Y}_{\mathrm{1bit}}]_{m,p}\,[\mathbf{R}(\tilde{\mathbf{S}})]_{m,p}\ge 0,
\qquad \forall\, m,p.
\end{equation}
When $[\mathbf{R}(\tilde{\mathbf{S}})]_{m,p}\neq 0$, \eqref{eq:sign_consistency} is equivalent to having
no sign mismatches in \eqref{eq:01_loss_equiv}; the equality case corresponds to the threshold boundary
$|[\mathbf{A}(\Theta)\tilde{\mathbf{S}}]_{m,p}|=\tau_m$.

To obtain a smooth surrogate for the sign-consistency requirement
in \eqref{eq:sign_consistency}, we replace the $0$--$1$ mismatch count by a logistic surrogate.
For each entry $(m,p)$, define
\begin{equation}
\label{eq:logloss_scalar}
[\boldsymbol{\ell}(\tilde{\mathbf{S}})]_{m,p}
\triangleq
\log\!\Big(1+\exp\!\big(-\beta\,[\mathbf{Y}_{\mathrm{1bit}}]_{m,p}\,[\mathbf{R}(\tilde{\mathbf{S}})]_{m,p}\big)\Big),
\end{equation}
where $\beta>0$ controls the sharpness around the threshold. 

Aggregating over all entries yields the sign-consistency loss
\begin{equation}
\label{eq:L_matrix}
L(\tilde{\mathbf{S}})
\triangleq
\frac{1}{MP}\sum_{m=1}^{M}\sum_{p=1}^{P} [\boldsymbol{\ell}(\tilde{\mathbf{S}})]_{m,p}.
\end{equation}

The loss $L(\tilde{\mathbf{S}})$ depends on the magnitudes
$\big|[\mathbf{A}(\Theta)\tilde{\mathbf{S}}]_{m,p}\big|$, which are
non-differentiable when $[\mathbf{A}(\Theta)\tilde{\mathbf{S}}]_{m,p}=0$.
To obtain a smooth objective, we introduce the smoothed magnitude 
applied to $\mathbf{A}(\Theta)\tilde{\mathbf{S}}$:
\begin{equation}
\label{eq:smoothed_mag_matrix}
\big[\,|\mathbf{A}(\Theta)\tilde{\mathbf{S}}|_\epsilon\,\big]_{m,p}
\triangleq
\sqrt{\big|[\mathbf{A}(\Theta)\tilde{\mathbf{S}}]_{m,p}\big|^2+\epsilon^2}.
\end{equation}
Here $\epsilon$ is a small smoothing parameter to regularize the magnitude
operation and ensure differentiability at zero.

Accordingly, define the smoothed threshold residual
\begin{equation}
\label{eq:residual_eps}
\mathbf{R}_\epsilon(\tilde{\mathbf{S}})
\triangleq
|\mathbf{A}(\Theta)\tilde{\mathbf{S}}|_\epsilon-\mathbf{T},
\end{equation}
and the smooth sign-consistency loss
\begin{equation}
\label{eq:L_eps}
L_\epsilon(\tilde{\mathbf{S}})
\triangleq
\frac{1}{MP}\sum_{m=1}^{M}\sum_{p=1}^{P}
\log\!\left(1+e^{-\beta\,[\mathbf{Y}_{\mathrm{1bit}}]_{m,p}\,
[\mathbf{R}_\epsilon(\tilde{\mathbf{S}})]_{m,p}}\right).
\end{equation}

Combining \eqref{eq:L_eps} with an $\ell_{2,1}$ penalty yields the proposed
group sparsity formulation
\begin{equation}
\label{eq:main_obj_matrix}
\min_{\tilde{\mathbf{S}}\in\mathbb{C}^{G\times P}}
\;F_\epsilon(\tilde{\mathbf{S}})
\triangleq
L_\epsilon(\tilde{\mathbf{S}})+\eta\|\tilde{\mathbf{S}}\|_{2,1},
\end{equation}
where $\eta>0$ and $\|\tilde{\mathbf{S}}\|_{2,1}=\sum_{g=1}^{G}\|\tilde{\mathbf{S}}_{g,:}\|_2$.
For reference, we also define the unsmoothed objective
\begin{equation}
\label{eq:F_unsmoothed}
F(\tilde{\mathbf{S}})
\triangleq
L(\tilde{\mathbf{S}})+\eta\|\tilde{\mathbf{S}}\|_{2,1}.
\end{equation}
Problem \eqref{eq:main_obj_matrix} is tailored to one-bit observations, in that 
$L_\epsilon$ depends only on the threshold-side information in
$\mathbf{Y}_{\mathrm{1bit}}$ and avoids fitting unavailable full-precision magnitudes.

\medskip
The loss $L_\epsilon(\tilde{\mathbf{S}})$ in \eqref{eq:L_eps} is continuously differentiable.
Since it is real-valued and depends on complex variables,  Wirtinger calculus is employed here.
For notational convenience, let $\mathbf{Z} = \mathbf{A}(\Theta)\tilde{\mathbf{S}}\in\mathbb{C}^{M\times P}$ denote the noiseless array output.
For each entry $(m,p)$, the partial derivative of 
$\log(1+e^{-\beta\,[\mathbf{Y}_{\mathrm{1bit}}]_{m,p}(|Z_{m,p}|_\epsilon-\tau_m)})$
with respect to $Z_{m,p}^*$ is
\begin{align}
\label{eq:partial_ell}
\frac{\partial \ell_{m,p}}{\partial Z_{m,p}^*}
&=
-\frac{\beta\,[\mathbf{Y}_{\mathrm{1bit}}]_{m,p}}{2|Z_{m,p}|_\epsilon} \nonumber\\
&\quad \cdot \sigma\!\big(-\beta\,[\mathbf{Y}_{\mathrm{1bit}}]_{m,p}(|Z_{m,p}|_\epsilon-\tau_m)\big)
\cdot Z_{m,p},
\end{align}
where $\sigma(t)=(1+e^{-t})^{-1}$ denotes the sigmoid function and $|Z_{m,p}|_\epsilon=\sqrt{|Z_{m,p}|^2+\epsilon^2}$.
We adopt the standard convention for Wirtinger derivatives of real-valued functions as $\nabla_{\mathbf{Z}}f \triangleq 2 \frac{\partial f}{\partial \mathbf{Z}^*}$. This definition ensures that the complex gradient corresponds directly to the direction of steepest ascent in the underlying real vector space.

Collecting \eqref{eq:partial_ell} over $(m,p)$ and applying the chain rule through
$\mathbf{Z}=\mathbf{A}(\Theta)\tilde{\mathbf{S}}$, the gradient of
$L_\epsilon=\frac{1}{MP}\sum_{m=1}^{M}\sum_{p=1}^{P}\ell_{m,p}$ can be written as
\begin{equation}
\label{eq:grad_L_eps}
\nabla L_\epsilon(\tilde{\mathbf{S}})
=
\frac{1}{MP}\mathbf{A}(\Theta)^{\mathsf H}\big(\mathbf{P}_\epsilon \odot \mathbf{W}\big),
\end{equation}
and
\begin{align}
\label{eq:P_eps_def}
[\mathbf{P}_\epsilon]_{m,p}
&\triangleq \frac{Z_{m,p}}{|Z_{m,p}|_\epsilon},\\
\label{eq:W_def}
[\mathbf{W}]_{m,p}
&\triangleq
-\beta\,[\mathbf{Y}_{\mathrm{1bit}}]_{m,p}\,
\sigma\!\Big(-\beta\,[\mathbf{Y}_{\mathrm{1bit}}]_{m,p}\,
[\mathbf{R}_\epsilon(\tilde{\mathbf{S}})]_{m,p}\Big),
\end{align}
with $\mathbf{R}_\epsilon(\tilde{\mathbf{S}})=|\mathbf{Z}|_\epsilon-\mathbf{T}$ as in
\eqref{eq:residual_eps}.

\begin{lemma}
\label{lem:Lipschitz}
The gradient $\nabla L_\epsilon$ is Lipschitz continuous with constant
\begin{equation}
\label{eq:Lip_const}
L_{\mathrm{Lip}}
=
\frac{\beta\|\mathbf{A}(\Theta)\|_2^2}{MP}
\left(\frac{\beta}{4} + \frac{1}{\epsilon}\right).
\end{equation}
\end{lemma}

\begin{proof}
See Appendix~\ref{app:proof_Lipschitz}.
\end{proof}

The constant $L_{\mathrm{Lip}}$ reflects the combined curvature of the logistic loss and the smoothed magnitude. Specifically, the term $\beta^2/4$ bounds the second derivative of the logistic function, while $\beta/\epsilon$ arises from the Jacobian bound of the smoothed magnitude operator $|z|_\epsilon$. As $\epsilon \to 0$, the gradient changes more rapidly near zero, requiring a smaller step size.

\begin{lemma}
\label{lem:smoothing_error}
For any $\tilde{\mathbf{S}}\in\mathbb{C}^{G\times P}$,
$|L_\epsilon(\tilde{\mathbf{S}}) - L(\tilde{\mathbf{S}})| \le \beta\epsilon$.
\end{lemma}

The proof is provided in Appendix~\ref{app:proof_smoothing_error}. This uniform bound implies that the gap between the optimal objective values of the smoothed and unsmoothed problems is bounded by $\beta\epsilon$, i.e., $|F_\epsilon(\tilde{\mathbf{S}}_\epsilon^\star) - F(\tilde{\mathbf{S}}^\star)| \le \beta\epsilon$. This justifies using a moderate $\epsilon$ (e.g., $10^{-3}$) to balance accuracy and conditioning.

\begin{remark}
As $\epsilon\to 0$, the smoothed magnitude $|z|_\epsilon\to|z|$ and the surrogate objective $L_\epsilon$ recovers the original non-smooth loss $L$, but the Lipschitz constant $L_{\mathrm{Lip}}$ in~\eqref{eq:Lip_const} grows as $\mathcal{O}(1/\epsilon)$, requiring smaller step sizes and hence more iterations for convergence. Conversely, a larger $\epsilon$ yields a better-conditioned landscape with faster convergence, but introduces an approximation bias bounded by $\beta\epsilon$ (Lemma~\ref{lem:smoothing_error}).
\end{remark}

With the gradient \eqref{eq:grad_L_eps} and the Lipschitz bound
\eqref{eq:Lip_const}, we minimize \eqref{eq:main_obj_matrix} using a
proximal-gradient scheme. Starting from $\tilde{\mathbf{S}}^{(0)}$, the update is
\begin{equation}
\label{eq:pg_update}
\tilde{\mathbf{S}}^{(t+1)}
=
\prox_{\mu\eta\|\cdot\|_{2,1}}
\!\Big(\tilde{\mathbf{S}}^{(t)}-\mu\,\nabla L_\epsilon(\tilde{\mathbf{S}}^{(t)})\Big),
\end{equation}
where the step size satisfies $\mu<1/L_{\mathrm{Lip}}$ (e.g., $\mu=c/L_{\mathrm{Lip}}$ with
$c\in(0,1)$). The proximal operator of the $\ell_{2,1}$ norm is the  block
soft-thresholding~\cite{parikh2014proximal}:
\begin{equation}
\label{eq:block_shrink}
\big[\prox_{\gamma\|\cdot\|_{2,1}}(\mathbf{X})\big]_{g,:}
=
\left(1-\frac{\gamma}{\|\mathbf{X}_{g,:}\|_2}\right)_{\!+}\mathbf{X}_{g,:},
\end{equation}
for $g=1,\ldots,G$, where $(a)_+ = \max\{a,\,0\}$ denotes the positive part of $a$. Hence an entire row is set to zero when
$\|\mathbf{X}_{g,:}\|_2\le\gamma$, which promotes row sparsity across multiple snapshots.

Since $F_\epsilon$ is nonconvex, we use multiple random initializations and select the
run attaining the smallest objective value. Specifically, we draw
$\tilde{\mathbf{S}}^{(0)}\sim\mathcal{CN}(\mathbf{0},\sigma_0^2\mathbf{I})$,
run \eqref{eq:pg_update} until the relative change of $F_\epsilon$ falls below a tolerance
$\delta_{\mathrm{tol}}$ (or $t=T_{\max}$), and repeat this for $N_{\mathrm{init}}$ times.
After selecting $\tilde{\mathbf{S}}^{\star}$, DOAs are obtained by taking the $K$ grid angles
corresponding to the largest row norms $\|\tilde{\mathbf{S}}^{\star}_{g,:}\|_2$.

\begin{remark}
\label{rem:initialization}
The objective $F_\epsilon$ is nonconvex because, for each entry with 
$[\mathbf{Y}_{\mathrm{1bit}}]_{m,p}=+1$, 
$\log\!\bigl(1+e^{-\beta(|z|_\epsilon-\tau_m)}\bigr)$ is the composition 
of a convex nonincreasing function with the convex map $z\mapsto|z|_\epsilon$, 
which does not preserve convexity. 
To mitigate the resulting local minima, we employ 
$N_{\mathrm{init}}$ independent random restarts and retain 
$\tilde{\mathbf{S}}^{\star} = \arg\min_{1\le n\le N_{\mathrm{init}}} 
F_\epsilon(\tilde{\mathbf{S}}_n)$ as described in Algorithm~\ref{alg:proposed}.
\end{remark}

The proposed One-BIt Magnitude-Only DOA Estimation via
Sign-consistency and group sparsiTy (OBI-MODEST) algorithm is  summarized in Algorithm~\ref{alg:proposed}.

\begin{algorithm}[t]
\caption{The Proposed OBI-MODEST}
\label{alg:proposed}
\begin{algorithmic}[1]
\Require One-bit observations $\mathbf{Y}_{\mathrm{1bit}}\in\{-1,+1\}^{M\times P}$,
steering matrix $\mathbf{A}(\Theta)\in\mathbb{C}^{M\times G}$,
thresholds $\boldsymbol{\tau}\in\mathbb{R}^{M}$,
and parameters $(\beta,\eta,\epsilon,\mu,\delta_{\mathrm{tol}},T_{\max},N_{\mathrm{init}},\sigma_0)$
\Ensure DOA estimates $\hat{\boldsymbol{\theta}}$
\State $\mathbf{T} \gets \boldsymbol{\tau}\mathbf{1}_P^{\mathsf T}$ 

\For{$n = 1, \ldots, N_{\mathrm{init}}$}
    \State Initialize $\tilde{\mathbf{S}}^{(0)} \sim \mathcal{CN}(\mathbf{0}, \sigma_0^2\mathbf{I})$
    \For{$t = 0, 1, \ldots, T_{\max}-1$}
        \State $\mathbf{Z} \gets \mathbf{A}(\Theta)\tilde{\mathbf{S}}^{(t)}$ 
        \State Compute $|\mathbf{Z}|_\epsilon$ via \eqref{eq:smoothed_mag_matrix}
        \State Compute $\mathbf{R}_\epsilon$, $\mathbf{P}_\epsilon$, $\mathbf{W}$ via \eqref{eq:residual_eps}, \eqref{eq:P_eps_def}, \eqref{eq:W_def}
        \State Compute $\nabla L_\epsilon(\tilde{\mathbf{S}}^{(t)})$ via \eqref{eq:grad_L_eps}
        \State Update $\tilde{\mathbf{S}}^{(t+1)}$ via \eqref{eq:pg_update}--\eqref{eq:block_shrink}
        \State $\Delta_{\mathrm{rel}} \gets \big|F_\epsilon^{(t+1)}-F_\epsilon^{(t)}\big|\big/F_\epsilon^{(t)}$ 
        \If{$\Delta_{\mathrm{rel}} < \delta_{\mathrm{tol}}$} \textbf{break} \EndIf
    \EndFor
    \State Store $(\tilde{\mathbf{S}}_n,\,F_\epsilon(\tilde{\mathbf{S}}_n))$
\EndFor

\State $\tilde{\mathbf{S}}^{\star} \gets \arg\min_n F_\epsilon(\tilde{\mathbf{S}}_n)$
\State \Return $\hat{\boldsymbol{\theta}} \gets K$ grid angles with largest $\|\tilde{\mathbf{S}}^{\star}_{g,:}\|_2$
\end{algorithmic}
\end{algorithm}

\subsection{Convergence Analysis}
\label{subsec:convergence}

We now analyze the convergence properties of OBI-MODEST for minimizing the composite objective $F_\epsilon(\tilde{\mathbf{S}}) = L_\epsilon(\tilde{\mathbf{S}}) + \eta\|\tilde{\mathbf{S}}\|_{2,1}$. Although the scalar logistic loss is convex, its composition with the magnitude of a complex linear map renders $L_\epsilon$ generally nonconvex with respect to $\tilde{\mathbf{S}}$. Consequently, global optimality cannot be guaranteed; instead, we establish monotone descent and convergence to a critical point, as formalized in the following Lemma~\ref{lem:descent} and Theorem~\ref{thm:conv_critical}.

\begin{lemma}
\label{lem:descent}
Let $\{\tilde{\mathbf{S}}^{(t)}\}_{t\ge 0}$ be the sequence generated by OBI-MODEST with step size $\mu \in (0, 1/L_{\mathrm{Lip}})$. Then, for all $t\ge 0$,
\begin{equation}
\label{eq:descent_ineq}
F_\epsilon(\tilde{\mathbf{S}}^{(t+1)}) \le F_\epsilon(\tilde{\mathbf{S}}^{(t)}) - \frac{1-\mu L_{\mathrm{Lip}}}{2\mu}\|\tilde{\mathbf{S}}^{(t+1)}-\tilde{\mathbf{S}}^{(t)}\|_F^2.
\end{equation}
In particular, $\{F_\epsilon(\tilde{\mathbf{S}}^{(t)})\}$ is monotonically nonincreasing.
\end{lemma}

To ensure the existence of accumulation points, we consider the initial sublevel set
\begin{equation}
\label{eq:sublevel_set}
\mathcal{S}_0 \triangleq \bigl\{\tilde{\mathbf{S}} \in \mathbb{C}^{G \times P} : F_\epsilon(\tilde{\mathbf{S}}) \le F_\epsilon(\tilde{\mathbf{S}}^{(0)})\bigr\}.
\end{equation}
By Lemma~\ref{lem:descent}, all iterations satisfy $\tilde{\mathbf{S}}^{(t)} \in \mathcal{S}_0$. The boundedness of $\mathcal{S}_0$ is guaranteed by the coercivity of the objective. Specifically, since $L_\epsilon(\tilde{\mathbf{S}}) \ge 0$, we have
\begin{align}
\label{eq:l21_coercive}
F_\epsilon(\tilde{\mathbf{S}}) 
&\ge \eta\|\tilde{\mathbf{S}}\|_{2,1} 
= \eta\sum_{g=1}^{G}\|\tilde{\mathbf{S}}_{g,:}\|_2 \nonumber\\
&\ge \eta \sqrt{\sum_{g=1}^{G}\|\tilde{\mathbf{S}}_{g,:}\|_2^2} 
= \eta \|\tilde{\mathbf{S}}\|_F,
\end{align}
which implies $F_\epsilon(\tilde{\mathbf{S}}) \to \infty$ as $\|\tilde{\mathbf{S}}\|_F \to \infty$. Thus, $\mathcal{S}_0$ is bounded.

\begin{theorem}
\label{thm:conv_critical}
Let $\{\tilde{\mathbf{S}}^{(t)}\}_{t\ge 0}$ be generated by OBI-MODEST with step size $\mu\in(0,1/L_{\mathrm{Lip}})$. Then, the successive differences vanish, i.e., $\|\tilde{\mathbf{S}}^{(t+1)}-\tilde{\mathbf{S}}^{(t)}\|_F\to 0$ as $t\to\infty$. The sequence admits at least one accumulation point, and every accumulation point $\tilde{\mathbf{S}}^\star$ satisfies the stationarity condition
\begin{equation}
\mathbf{0} \in \nabla L_\epsilon(\tilde{\mathbf{S}}^\star) + \partial\bigl(\eta\|\cdot\|_{2,1}\bigr)(\tilde{\mathbf{S}}^\star),
\end{equation}
meaning that $\tilde{\mathbf{S}}^\star$ is a critical point of $F_\epsilon$.
\end{theorem}

The proofs of Lemma~\ref{lem:descent} and Theorem~\ref{thm:conv_critical} are provided in Appendices~\ref{app:proof_descent} and~\ref{app:proof_convergence}, respectively.

\subsection{Computational Complexity}
\label{subsec:complexity}

We analyze the computational complexity of the proposed OBI-MODEST algorithm and compare it with three 
methods: CBIHT-$\ell_2$~\cite{stockle2015spawc}, ToyBar~\cite{wan2021dual}, and 
GESPAR~\cite{kim2014noncoherent}. Let $M$ denote the number of sensors, $G$ the grid size, 
$P$ the number of snapshots, and $K$ the number of sources.

The per-iteration cost of OBI-MODEST is dominated by the forward transform 
$\mathbf{Z} = \mathbf{A}(\Theta)\tilde{\mathbf{S}}$ and the adjoint transform 
$\mathbf{A}(\Theta)^{\mathsf{H}}(\mathbf{P}_\epsilon \odot \mathbf{W})$, each requiring 
$\mathcal{O}(MGP)$ complex multiply-add operations. The block soft-thresholding step in \eqref{eq:block_shrink} costs $\mathcal{O}(GP)$.
All other computations are element-wise, including the formation of $|\mathbf{Z}|_\epsilon$ in \eqref{eq:smoothed_mag_matrix}, $\mathbf{R}_\epsilon$ in \eqref{eq:residual_eps}, and $(\mathbf{P}_\epsilon,\mathbf{W})$ in \eqref{eq:P_eps_def}--\eqref{eq:W_def}, as well as the evaluation of $L_\epsilon$ in \eqref{eq:L_eps}; their total cost scales as $\mathcal{O}(MP)$, which is a  lower order compared with $\mathcal{O}(MGP)$.
Since the objective $F_\epsilon$ is nonconvex, we employ $N_{\mathrm{init}}$ independent random 
initializations and select the solution with the smallest objective value. If each run terminates 
within $T_{\max}$ iterations, the total complexity of the proposed algorithm is
\begin{equation}
\label{eq:complexity_ours}
\mathcal{O}\bigl(N_{\mathrm{init}} \cdot T_{\max} \cdot MGP\bigr).
\end{equation}

Each iteration of CBIHT-$\ell_2$~\cite{stockle2015spawc} is dominated by two matrix-vector multiplications: one is the forward computation $\mathbf{A}\hat{\mathbf{S}}$ and the other is the backward gradient step $\mathbf{A}^{\mathsf{H}}(\mathbf{Y}_{\mathrm{1bit}} - \operatorname{sign}(\mathbf{A}\hat{\mathbf{S}}))$, both costing $\mathcal{O}(MGP)$.
A key advantage of CBIHT-$\ell_2$ is that it requires no random restarts; instead, it uses a deterministic initialization via back-projection of the measurements, i.e., $\hat{\mathbf{S}}^{(0)} = \mathbf{A}^{\mathsf{H}}\mathbf{Y}_{\mathrm{1bit}}$.
Since it typically converges in a small number of iterations $T_C$, the total complexity is $\mathcal{O}(T_C \cdot MGP)$, which effectively scales as $\mathcal{O}(MGP)$.

ToyBar~\cite{wan2021dual} solves a majorization-minimization surrogate via proximal 
gradient with Nesterov acceleration. Its per-iteration structure is similar to the proposed algorithm, 
costing $\mathcal{O}(MGP)$. However, ToyBar also addresses a nonconvex problem and thus 
requires multiple random initializations. The total complexity is therefore 
$\mathcal{O}(N_T \cdot T_T \cdot MGP)$.

GESPAR~\cite{kim2014noncoherent} is a greedy algorithm that alternates between support 
refinement via local search and coefficient estimation via Damped Gauss-Newton (DGN). 
For multiple snapshots, the coefficients are estimated per column, requiring $P$ 
calls to DGN solver  per outer iteration. Each DGN step involves forming and solving a 
$(K + K_{\mathrm{ref}}) \times (K + K_{\mathrm{ref}})$ linear system, costing 
$\mathcal{O}(K^3)$, plus residual evaluations costing $\mathcal{O}(MK)$. The local search 
phase considers candidate support swaps by examining $\mathcal{O}(GK)$ index pairs. 
As multiple random initializations 
($N_G$ restarts) are  required, with a maximum of $T_G$ iterations per restart, 
the total complexity scales as
\begin{equation}
\mathcal{O}\bigl(N_G \cdot T_G \cdot (K^3 P + MKP + GK)\bigr).
\end{equation}
In practice, the $GK$ term from local search is often dominated by the $MKP$ term from 
residual evaluations when $M P \gg G$, which is typical in DOA estimation where many 
snapshots are collected. However, when $G$ is large or the number 
of snapshots is small, the local search overhead becomes significant.

Table~\ref{tab:complexity} summarizes the results. Although the proposed OBI-MODEST algorithm, ToyBar, and CBIHT-$\ell_2$ 
share the same $\mathcal{O}(MGP)$ per-iteration cost, their total complexities differ 
significantly due to the number of iterations and random restarts required. GESPAR incurs the highest cost due to its greedy search and 
per-snapshot DGN solver .

\begin{table}[t]
\centering
\caption{Computational Complexity Comparison}
\label{tab:complexity}
\footnotesize
\setlength{\tabcolsep}{3pt}
\begin{tabular}{l|c|c|c}
\toprule
\textbf{Method} & \textbf{Per-Iter.} & \textbf{Restarts} & \textbf{Total Complexity} \\
\midrule
OBI-MODEST & $\mathcal{O}(MGP)$ & $N_{\mathrm{init}}$ & 
$\mathcal{O}(N_{\mathrm{init}} T_{\max} MGP)$ \\
CBIHT-$\ell_2$~\cite{stockle2015spawc} & $\mathcal{O}(MGP)$ & 1 & 
$\mathcal{O}(T_C MGP)$ \\
ToyBar~\cite{wan2021dual} & $\mathcal{O}(MGP)$ & $N_T$ & 
$\mathcal{O}(N_T T_T MGP)$ \\
GESPAR~\cite{kim2014noncoherent} & $\mathcal{O}(K^3\!+\!MKP)$ & $N_G$ & 
$\mathcal{O}(N_G T_G (K^3\!+\!MK) P)$ \\
\bottomrule
\end{tabular}
\vspace{1mm}
\begin{flushleft}
\scriptsize
$N_T$, $N_G$: number of restarts for ToyBar and GESPAR, respectively.\\
$T_C$, $T_T$, $T_G$: number of iterations for CBIHT-$\ell_2$, ToyBar, and GESPAR.
\end{flushleft}
\end{table}

\section{Simulation Results}
\label{sec:simulation}

\subsection{Simulation Setup}
\label{subsec:sim_setup}

All simulations employ a 19-sensor UCA.
The angular search grid spans $[-90^\circ,90^\circ]$ with $361$ grid points.
The sensor phase errors are modeled by $\mathbf{D}=\diag(e^{j\phi_1},\ldots,e^{j\phi_M})$ with $\phi_m\sim\mathcal{N}(0,\sigma_\phi^2)$.
The signal-to-noise ratio (SNR) is defined as $\mathrm{SNR} = 10 \log_{10} ( \sigma_s^2 / \sigma_n^2 )$.

Since the algorithms under consideration are developed based on different signal models, we introduce their respective measurement formulations below.

The proposed method operates on the one-bit magnitude-only measurements with unknown phase errors, i.e., 
$\mathbf{Y}_{\mathrm{1bit}} = \operatorname{sign}\big(|\mathbf{D}\mathbf{A}(\Theta)\tilde{\mathbf{S}}| + \mathbf{N} - \mathbf{T}\big)$.
The coherent one-bit baselines (CBIHT-$\ell_2$~\cite{stockle2015spawc} and ROCS~\cite{li2024rocs}) quantize the complex baseband samples directly as $\mathbf{Y}^{\mathrm{coh}}_{\mathrm{1bit}}
=
\operatorname{sign}_{\mathbb{C}}\!\big(\mathbf{D}\mathbf{A}(\Theta)\tilde{\mathbf{S}}
+\mathbf{N}_{\mathrm{c}}\big)$,
where $\mathbf{N}_{\mathrm{c}} \in \mathbb{C}^{M \times P}$ represents complex noise, and $\operatorname{sign}_{\mathbb{C}}(\cdot) = \operatorname{sign}(\Re\{\cdot\}) + j\operatorname{sign}(\Im\{\cdot\})$.

For the unquantized benchmarks, GESPAR~\cite{kim2014noncoherent} utilizes intensity data $\mathbf{Y}_{\mathrm{int}} = |\mathbf{D}\mathbf{A}(\Theta)\tilde{\mathbf{S}} + \mathbf{N}_{\mathrm{c}}|^{2}$, while ToyBar~\cite{wan2021dual} employs magnitude data $\mathbf{Y}_{\mathrm{mag}} = |\mathbf{D}\mathbf{A}(\Theta)\tilde{\mathbf{S}}| + \mathbf{N}$. Finally, the unquantized coherent methods (MUSIC, one-bit MUSIC~\cite{huang2019one}, and Eigenstructure~\cite{liu2011eigenstructure}) operate on the standard measurements $\mathbf{Y}_{\mathrm{c}} = \mathbf{D}\mathbf{A}(\Theta)\tilde{\mathbf{S}} + \mathbf{N}_{\mathrm{c}}$.

For OBI-MODEST, we set the per-sensor threshold as 
$\tau_m = \mathrm{median}\bigl(\{[\mathbf{Y}]_{m,p}\}_{p=1}^{P}\bigr)$. The logistic sharpness and smoothing parameters are fixed at $\beta = 2$ and $\epsilon = 10^{-3}$, respectively.  The step size is set to a conservative value of $\mu = 0.25/L_{\mathrm{Lip}}$, where $L_{\mathrm{Lip}}$ is computed via~\eqref{eq:Lip_const}.
Regarding initialization, the starting point $\tilde{\mathbf{S}}^{(0)}$ is initialized with $\tilde{\mathbf{S}}^{(0)} \sim \mathcal{CN}(\mathbf{0}, 0.1^{2}\mathbf{I})$. The algorithm terminates when the relative change in the objective function falls below $\delta_{\mathrm{tol}} = 10^{-4}$ or the iteration count reaches $T_{\max} = 400$. Finally, to mitigate sensitivity to local minima, we employ $N_{\mathrm{init}} = 5$ independent random initializations and select the solution with the smallest final objective value.

For baseline methods, we use the empirically tuned parameters. CBIHT-$\ell_2$ runs for 10 iterations with deterministic initialization~\cite{stockle2015spawc}; ROCS uses the tanh-smoothing parameter
$c_{\mathrm{R}}=1$, proximal parameter $\mu_{\mathrm{R}} = 0.05$, step size $\eta_{\mathrm{R}} = 0.005$, and penalty $\gamma_{\mathrm{R}} = 4$, with 50 outer and 30 inner iterations~\cite{li2024rocs}; ToyBar and GESPAR use 5 random restarts with 200 iterations each.

Performance is evaluated by the root mean square error (RMSE) in degrees:
\begin{equation}
\mathrm{RMSE}
=
\sqrt{\frac{1}{JK}\sum_{j=1}^{J}\sum_{k=1}^{K}
\big(\hat{\theta}_{j,k}-\theta_k\big)^2},
\label{eq:rmse_def}
\end{equation}
where $J$ denotes the number of Monte Carlo trials and is chosen as 200 in corresponding simulations.

\subsection{Verification of Theoretical Analysis}
\label{subsec:theory_verification}

This subsection provides numerical validation of the theoretical results established in Sections~\ref{sec:problem_formulation} and~\ref{sec:proposed_method}, by considering $K=3$ sources at $\{-40^\circ, 0^\circ, 30.5^\circ\}$ with $P=100$ snapshots under $\mathrm{SNR}=20$~dB.

\subsubsection{Magnitude-Fitting versus Sign-Consistency}
\label{subsubsec:mag_vs_sign}

We first illustrate a fundamental mismatch between magnitude-fitting objectives and one-bit measurements, which motivates the proposed sign-consistency formulation.
Fix a true $K$-source coefficient matrix $\tilde{\mathbf{S}}_{\mathrm{true}}$ and form its noiseless magnitudes
$|\mathbf{A}\tilde{\mathbf{S}}_{\mathrm{true}}|$.
For each perturbation level $\delta\in[10^{-3},10^{1.5}]$, we generate
\begin{equation}
\tilde{\mathbf{S}}^{(\delta)}=\tilde{\mathbf{S}}_{\mathrm{true}}+\delta\,\Delta\tilde{\mathbf{S}},
\qquad
[\Delta\tilde{\mathbf{S}}]_{g,p}\sim\mathcal{CN}(0,1)\ \text{i.i.d.},
\end{equation}
and compute the corresponding one-bit outputs using the same thresholds
\begin{equation}
\mathbf{Y}_{\mathrm{1bit}}^{(\delta)}
=
\operatorname{sign}\!\Big(|\mathbf{A}\tilde{\mathbf{S}}^{(\delta)}|-\mathbf{T}\Big).
\end{equation}
We compare two metrics averaged over 200 random perturbations $\Delta\tilde{\mathbf{S}}$.
The first is the one-bit agreement rate
\begin{equation}
r(\delta) \triangleq \frac{1}{MP}\sum_{m,p}
\mathbf{1}\Big\{[\mathbf{Y}_{\mathrm{1bit}}^{(\delta)}]_{m,p}=
[\mathbf{Y}_{\mathrm{1bit}}^{\mathrm{true}}]_{m,p}\Big\},
\label{eq:r_delta}
\end{equation}
which directly measures sign consistency.
The second is the magnitude change
\begin{equation}
\|\Delta\mathbf{R}\|_F^2
\triangleq
\big\||\mathbf{A}\tilde{\mathbf{S}}^{(\delta)}|-|\mathbf{A}\tilde{\mathbf{S}}_{\mathrm{true}}|\big\|_F^2,
\label{eq:deltaR_def}
\end{equation}
which captures variations that are explicitly penalized by magnitude-fitting approaches (e.g., ToyBar), even when such variations do not alter the one-bit labels.

Fig.~\ref{fig:mag_vs_sign}(a) shows that small perturbations leave the one-bit labels nearly unchanged (e.g., $r(\delta)>0.96$ for $\delta<10^{-2}$), while $r(\delta)$ decays and eventually approaches $0.5$ as $\delta$ increases, corresponding to near-independent labels.
In contrast, Fig.~\ref{fig:mag_vs_sign}(b) demonstrates that $\|\Delta\mathbf{R}\|_F^2$ grows approximately on the order of $\delta^2$ even in the regime where $r(\delta)\approx 1$.
Therefore, magnitude-fitting objectives can still produce non-negligible descent directions that do not affect the one-bit measurements, potentially moving the iteration away from the  sign-consistent region.

By comparison, the proposed sign-consistency loss depends on whether the predicted magnitudes fall on the correct side of the thresholds.
Specifically, with $y_{m,p}\in\{+1,-1\}$, the desired condition for each entry is
\begin{equation}
y_{m,p}\Big(\big|[\mathbf{A}\tilde{\mathbf{S}}]_{m,p}\big|-\tau_m\Big) > 0.
\label{eq:threshold_side_condition}
\end{equation}
$\big|[\mathbf{A}\tilde{\mathbf{S}}]_{m,p}\big|>\tau_m$ when $y_{m,p}=+1$, and $\big|[\mathbf{A}\tilde{\mathbf{S}}]_{m,p}\big|<\tau_m$ when $y_{m,p}=-1$.
For the logistic loss $\ell(u)=\log(1+e^{-\beta u})$, one has $|\ell'(u)|=\beta/(1+e^{\beta u})$, which becomes negligible when \eqref{eq:threshold_side_condition} holds with a comfortable gap.
Consequently, the updates concentrate on entries whose predictions lie on the wrong side of the thresholds, rather than on magnitude variations that are invisible to the one-bit quantizer.

\begin{figure}[t]
\centering
\begin{minipage}[b]{0.48\columnwidth}
    \centering
    \includegraphics[width=\linewidth]{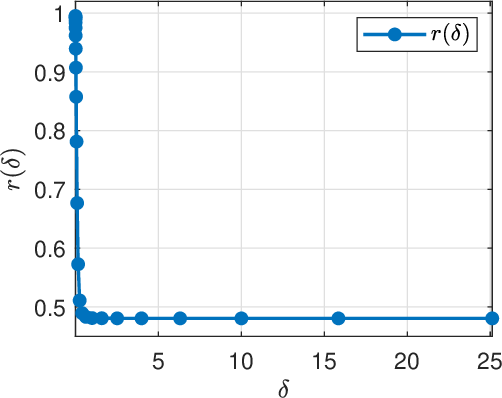}
    \centerline{\small (a)}
\end{minipage}
\hfill
\begin{minipage}[b]{0.48\columnwidth}
    \centering
    \includegraphics[width=\linewidth]{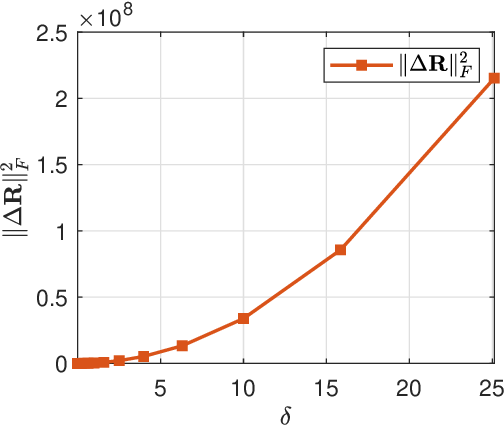}
    \centerline{\small (b)}
\end{minipage}
\caption{Mismatch between one-bit and magnitude objectives versus perturbation $\delta$:
(a) one-bit agreement rate $r(\delta)$;
(b) magnitude change $\|\Delta\mathbf{R}\|_F^2$.}
\label{fig:mag_vs_sign}
\end{figure}

\subsubsection{Convergence Behavior}
\label{subsubsec:convergence_exp}

We evaluate Lemma~\ref{lem:descent} and Theorem~\ref{thm:conv_critical} by tracking the objective
$F_\epsilon(\tilde{\mathbf{S}}^{(t)})$ together with three stationarity diagnostics:
\begin{align}
\|G_\mu^{(t)}\|_F
&=
\frac{1}{\mu}\big\|\tilde{\mathbf{S}}^{(t)}-\tilde{\mathbf{S}}^{(t+1)}\big\|_F,
\label{eq:pg_norm_def}\\
\kappa_{\max}^{(t)}
&=
\max_{g}\ \min_{\boldsymbol{\xi} \in \partial \|\tilde{\mathbf{S}}_{g,:}^{(t)}\|_2}
\big\| \nabla_g L_\epsilon(\tilde{\mathbf{S}}^{(t)}) + \eta \boldsymbol{\xi} \big\|_2,
\label{eq:kkt_def}\\
\delta_t
&=
\frac{\big\|\tilde{\mathbf{S}}^{(t)}-\tilde{\mathbf{S}}^{(t-1)}\big\|_F}{\big\|\tilde{\mathbf{S}}^{(t-1)}\big\|_F}.
\label{eq:delta_def}
\end{align}
Here $\|G_\mu^{(t)}\|_F$ is the normalized update step size,
$\kappa_{\max}^{(t)}$ measures the maximum  KKT residual for the objective
$L_\epsilon(\tilde{\mathbf{S}})+\eta\|\tilde{\mathbf{S}}\|_{2,1}$, and $\delta_t$ reports the relative iteration change.

Fig.~\ref{fig:convergence}(a) shows that $F_\epsilon(\tilde{\mathbf{S}}^{(t)})$ is monotonically nonincreasing over $1000$ iterations, which is consistent with the sufficient-decrease property in Lemma~\ref{lem:descent}.
In particular, Lemma~\ref{lem:descent} implies the descent inequality
\begin{equation}
F_\epsilon(\tilde{\mathbf{S}}^{(t)})-F_\epsilon(\tilde{\mathbf{S}}^{(t+1)})
\ \ge\ c\,\big\|\tilde{\mathbf{S}}^{(t+1)}-\tilde{\mathbf{S}}^{(t)}\big\|_F^2,
\label{eq:suff_decrease_discuss}
\end{equation}
for some constant $c>0$ determined by the stepsize condition, and thus 
$\|\tilde{\mathbf{S}}^{(t+1)}-\tilde{\mathbf{S}}^{(t)}\|_F\to 0$.

Fig.~\ref{fig:convergence}(b) further supports Theorem~\ref{thm:conv_critical}.
$\|G_\mu^{(t)}\|_F$ and $\delta_t$ decay to the $10^{-3}$--$10^{-4}$ level, indicating that successive updates vanish and the iterations stabilize,
while the concurrent reduction of $\kappa_{\max}^{(t)}$ indicates that the iterations approach a stationary point that approximately satisfies the
 KKT conditions associated with~\eqref{eq:main_obj_matrix}.

\begin{figure}[t]
\centering
\begin{minipage}[b]{0.48\columnwidth}
    \centering
    \includegraphics[width=\linewidth]{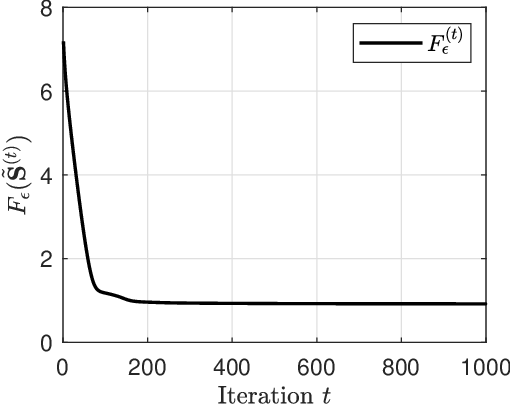}
    \centerline{\small (a)}
\end{minipage}
\hfill
\begin{minipage}[b]{0.48\columnwidth}
    \centering
    \includegraphics[width=\linewidth]{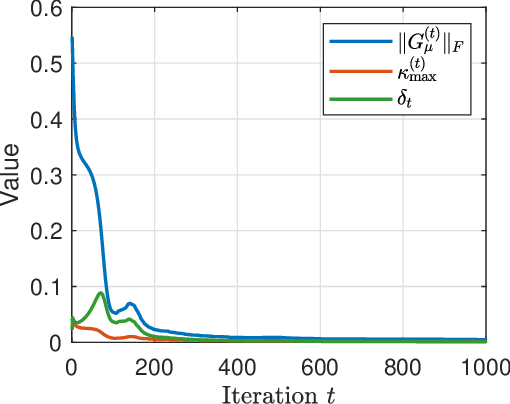}
    \centerline{\small (b)}
\end{minipage}
\caption{Convergence behavior:
(a) objective value $F_\epsilon(\tilde{\mathbf{S}}^{(t)})$ versus iteration;
(b) diagnostics $\|G_\mu^{(t)}\|_F$, $\kappa_{\max}^{(t)}$, and $\delta_t$.}
\label{fig:convergence}
\end{figure}

\subsection{Regularization Parameter Selection}
\label{subsec:param_selection}

The regularization parameter $\eta$ in~\eqref{eq:main_obj_matrix} serves to balance the logistic loss $L_\epsilon(\tilde{\mathbf{S}})$ against the row-sparsity penalty $\|\tilde{\mathbf{S}}\|_{2,1}$.
We examine the sensitivity of the algorithm to $\eta$ under varying SNR levels and snapshot numbers. The simulation setup involves $K=3$ sources located at $\{-40^\circ, 0^\circ, 30^\circ\}$ and
$\eta$ is selected via grid search.

Fig.~\ref{fig:eta_selection}(a) illustrates the RMSE as a function of $\eta$ for various SNRs with $P=100$, and
a characteristic U-shaped curve is observed. While small $\eta$ values lead to insufficient sparsity, excessively large $\eta$ values attenuate the signal of interest.
Notably, across the tested SNR range of $5$--$20$~dB, the empirical optimal $\eta$ falls consistently within the narrow interval $[0.023, 0.027]$. 

Fig.~\ref{fig:eta_selection}(b) reveals that the optimal $\eta$ is strongly dependent on the snapshot number $P$ (shown for $\mathrm{SNR}=15$~dB).
As $P$ increases from $50$ to $500$, the optimal parameter $\eta^\star$ decreases from $0.035$ to $0.012$.
To characterize this relationship, we analyze the product $\eta^\star\sqrt{P}$, which yields values of $0.247$, $0.250$, $0.255$, and $0.268$ for $P\in\{50, 100, 200, 500\}$, respectively.
The approximate constancy of this product motivates the adoption of the following scaling law:
\begin{equation}
\eta^\star \approx \frac{\alpha}{\sqrt{P}}, \quad \text{with } \alpha \approx 0.25.
\label{eq:eta_scaling}
\end{equation}

\begin{figure}[t]
\centering
\begin{minipage}[b]{0.48\columnwidth}
    \centering
    \includegraphics[width=\linewidth]{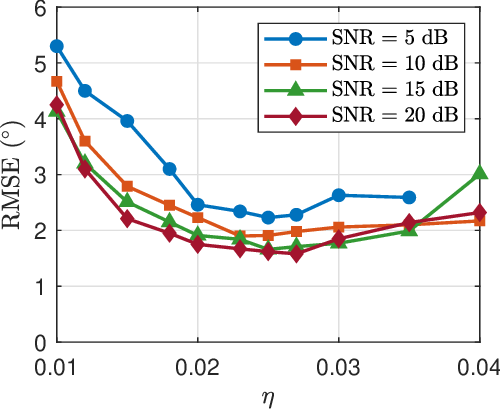}
    \centerline{\small (a)}
\end{minipage}
\hfill
\begin{minipage}[b]{0.48\columnwidth}
    \centering
    \includegraphics[width=\linewidth]{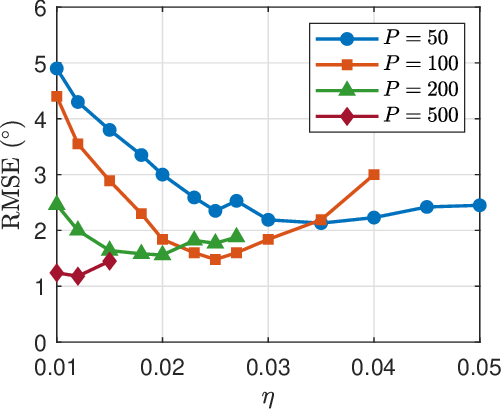}
    \centerline{\small (b)}
\end{minipage}
\caption{RMSE versus regularization parameter $\eta$:
(a) varying SNRs ($P=100$);
(b) varying snapshot counts $P$ ($\mathrm{SNR}=15$~dB).}
\label{fig:eta_selection}
\end{figure}

\subsection{DOA Estimation Performance}
\label{subsec:doa_estimation_performance}

We first evaluate the DOA estimation performance under ideal conditions without sensor phase errors, with $K = 3$ sources at $[-40.7^\circ, 0.8^\circ, 30.2^\circ]$.

\subsubsection{Effect of SNR}
\label{subsubsec:snr_effect}

Fig.~\ref{fig:snr_vs_rmse} presents the RMSE performance versus SNR ranging from 5~dB to 20~dB, with the snapshot number fixed at $P = 80$.

\begin{figure}[!t]
    \centering
    \includegraphics[width=0.85\columnwidth]{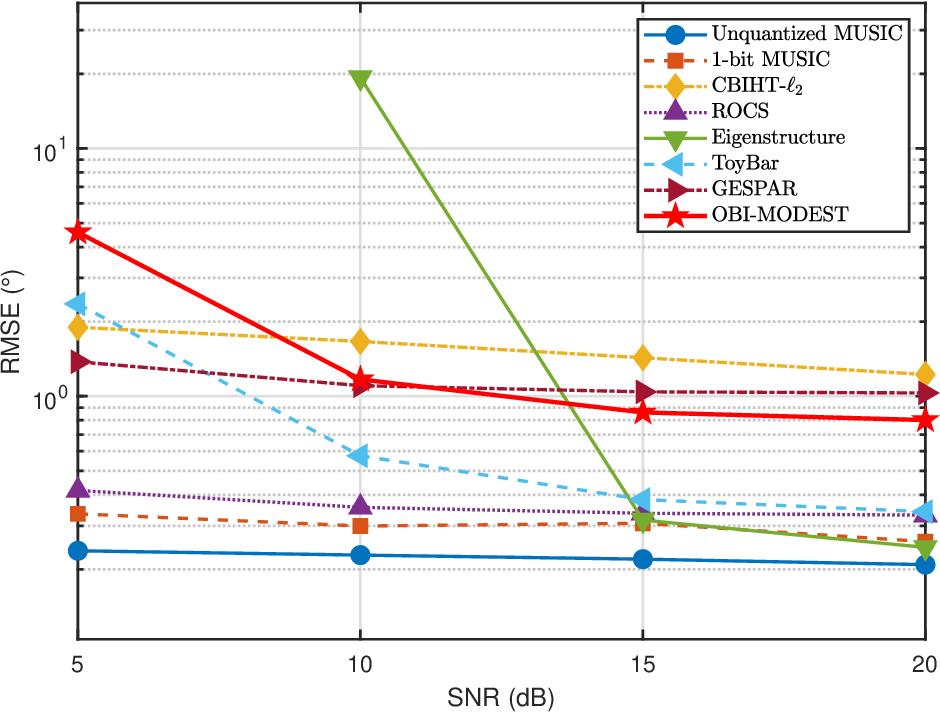}
    \caption{RMSE versus SNR with $P = 80$ snapshots and no phase errors.}
    \label{fig:snr_vs_rmse}
\end{figure}

The coherent benchmarks such as Unquantized MUSIC, 1-bit MUSIC, and ROCS benefit from the availability of phase information. Specifically, Unquantized MUSIC and 1-bit MUSIC maintain a low RMSE ($<0.4^\circ$) across the entire SNR range. However, interesting trends emerge among the magnitude-based/non-coherent methods.
The Eigenstructure method fails at low SNR ($5$--$10$~dB) due to its heavy reliance on high-quality statistical accumulation but converges rapidly to high accuracy at $20$~dB.
ToyBar utilizes unquantized magnitude measurement, and exhibits significant improvement as SNR increases, clearly outperforming the 1-bit coherent method CBIHT-$\ell_2$ at $\mathrm{SNR} \ge 10$~dB. This suggests that precise amplitude information can sometimes compensate for the lack of phase information, whereas the convex relaxation in CBIHT-$\ell_2$ limits its resolution at higher SNRs.

The proposed method demonstrates a moderate performance. Although it starts with a larger error at $5$~dB due to severe information loss inherent in 1-bit magnitude quantization, its RMSE drops sharply to $\approx 1.1^\circ$ at $10$~dB.
Crucially, for $\mathrm{SNR} \ge 10$~dB, the proposed method surpasses the coherent 1-bit baseline CBIHT-$\ell_2$.
Furthermore, despite using only coarse 1-bit magnitude measurements, our method achieves an accuracy comparable to, and at $20$~dB slightly superior to GESPAR, which utilizes full-precision magnitude measurements.
This indicates that the proposed method effectively mitigating the effect of information deficit caused by extreme quantization and phase loss.

\subsubsection{Effect of Number of Snapshots}
\label{subsubsec:snapshot_effect}

Fig.~\ref{fig:snapshots_vs_rmse} shows the RMSE versus the snapshot number $P\in\{20,50,80,100\}$ at $\mathrm{SNR}=15$~dB.

\begin{figure}[!t]
    \centering
    \includegraphics[width=0.85\columnwidth]{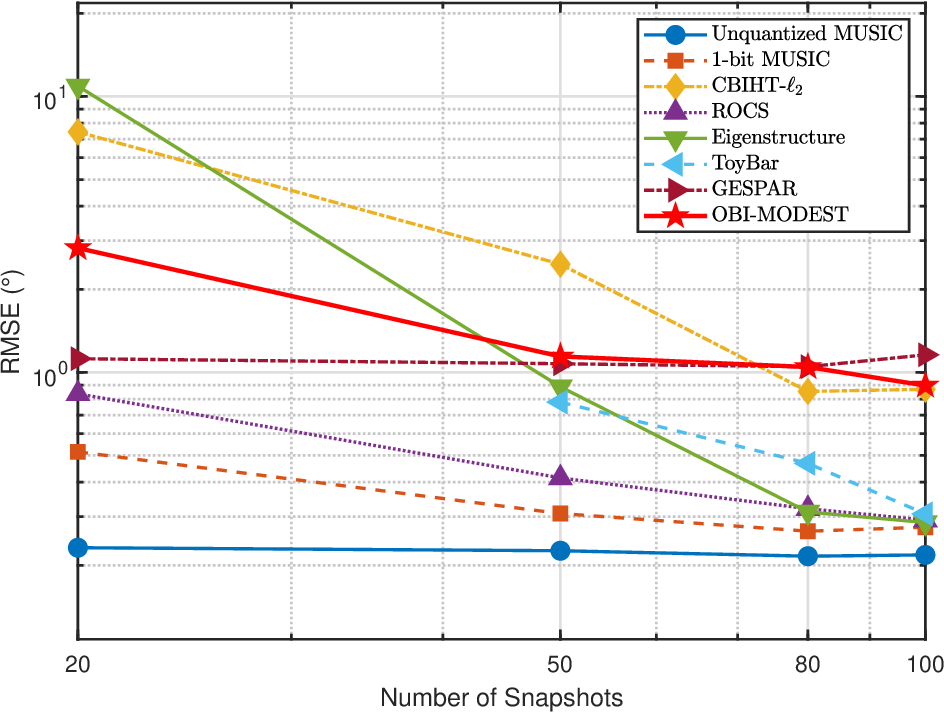}
    \caption{RMSE versus number of snapshots at $\mathrm{SNR}=15$~dB with no phase errors.}
    \label{fig:snapshots_vs_rmse}
\end{figure}

The Unquantized MUSIC has the best performance, while 1-bit MUSIC and ROCS gradually improve as $P$ increases.
In contrast, CBIHT-$\ell_2$ shows significant sensitivity to limited data, with its performance deteriorating to an RMSE above $7^\circ$ at $P=20$ and becomes competitive only when $P$ is sufficiently large, indicating a higher sample requirement for its convex recovery step in this configuration.

The Eigenstructure method performs poorly at $P=20$ but improves sharply once $P$ increases, consistent with the fact that its second-order-statistics construction is unreliable under limited averaging.
ToyBar does not return a valid estimate at $P=20$ and becomes effective only for $P\ge 50$, after which its RMSE decreases rapidly with $P$, reflecting the difficulty of nonconvex magnitude fitting.
GESPAR is comparatively stable across $P$ but exhibits weaker dependence on $P$.

The proposed method provides a meaningful estimate at $P=20$ with an RMSE of about $2.8^\circ$, where ToyBar fails completely and both CBIHT-$\ell_2$ and Eigenstructure incur significant errors.

\subsection{Phase Error Robustness}
\label{subsec:phase_error_robustness}

A primary advantage for adopting one-bit magnitude-only measurements is their robustness against phase errors. 
To validate this advantage numerically,  we fix $\mathrm{SNR}=15$~dB and $P=80$ snapshots, and vary the phase-error standard deviation
$\sigma_\phi$ from $0^\circ$ to $90^\circ$.

Fig.~\ref{fig:phase_error} plots the RMSE result versus $\sigma_\phi$, which reveals a distinct crossover in performance.
In the regime of small phase errors ($\sigma_\phi < 15^\circ$), coherent baselines (unquantized MUSIC, 1-bit MUSIC, and ROCS) indeed exhibit superior resolution compared to the proposed method, benefiting from relatively accurate  phase information.
However, as $\sigma_\phi$ increases, their performance degrades rapidly.
Notably, when $\sigma_\phi$ exceeds $30^\circ$, even unquantized MUSIC, which yields the lowest RMSE in the ideal case, deteriorates sufficiently to be outperformed by our proposed method. 

In  contrast, the proposed method exhibits a flat RMSE curve across the entire range of $\sigma_\phi$.
Remarkably, despite utilizing only coarse 1-bit magnitude quantization, our method achieves a performance comparable to that of GESPAR. This highlights a fundamental trade-off where coherent methods offer higher precision under perfect calibration, whereas the proposed approach ensures  robustness to phase errors while enabling a low-cost implementation. 

\begin{figure}[!t]
    \centering
    \includegraphics[width=0.85\columnwidth]{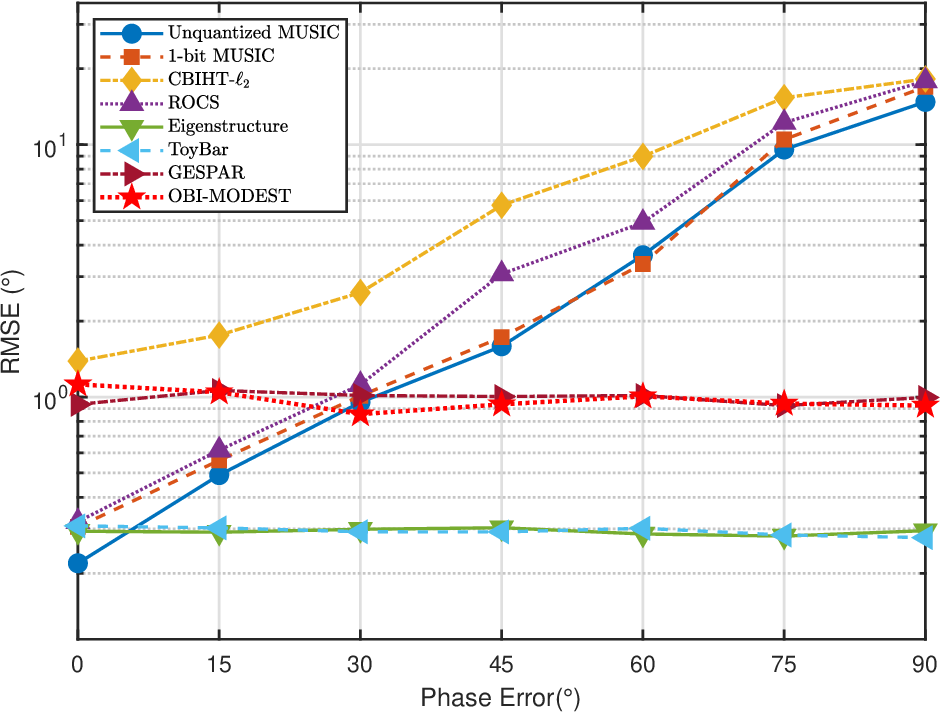}
    \caption{RMSE versus phase error standard deviation $\sigma_\phi$ at $\mathrm{SNR} = 15$~dB with $P = 80$ snapshots.}
    \label{fig:phase_error}
\end{figure}
\section{Conclusion}
\label{sec:conclusion}

The problem of DOA estimation from one-bit magnitude-only measurements has been addressed. Distinct from coherent estimators that degrade under phase errors and conventional non-coherent methods that rely on full-precision amplitude data, our proposed method formulates the problem as a sign-consistency optimization task using a smooth logistic surrogate to directly exploit binary threshold information. A proximal-gradient algorithm is developed with proven convergence to a critical point, specifically designed to handle the dual challenges of phase-less sensing and one-bit quantization. Numerical results demonstrate that the proposed method effectively bridges the gap between robustness and efficiency, and achieves accuracy comparable to coherent baselines under ideal conditions while delivering superior stability under severe phase errors, thereby establishing a practical solution for low-cost, uncalibrated array systems.

\appendices
\section{Proof of Lemma~\ref{lem:local_constancy}}
\label{app:proof_local_constancy}

\begin{proof}
By the non-degeneracy assumption in \eqref{eq:margin_condition}, the quantity $\bigl|[\mathbf{A}(\Theta)\tilde{\mathbf{S}}_0]_{m,p}\bigr|-\tau_m$ is nonzero for all $(m,p)$.
Recall the definition of the minimum decision margin $\rho$ from \eqref{eq:rho_def}:
\begin{equation}
\label{eq:rho_app}
\rho = \min_{m,p}
\left|\bigl|[\mathbf{A}(\Theta)\tilde{\mathbf{S}}_0]_{m,p}\bigr|-\tau_m\right| > 0.
\end{equation}

Now, assume that the magnitude perturbation satisfies \eqref{eq:local_delta_condition}, i.e.,
\begin{equation}
\bigl\||\mathbf{A}(\Theta)\tilde{\mathbf{S}}|-|\mathbf{A}(\Theta)\tilde{\mathbf{S}}_0|\bigr\|_\infty < \rho.
\end{equation}
For any index pair $(m,p)$, define the signed distance to the threshold for the unperturbed and perturbed signals, respectively, as:
\begin{align}
u_0 &\triangleq \bigl|[\mathbf{A}(\Theta)\tilde{\mathbf{S}}_0]_{m,p}\bigr| - \tau_m, \\
u   &\triangleq \bigl|[\mathbf{A}(\Theta)\tilde{\mathbf{S}}]_{m,p}\bigr| - \tau_m.
\end{align}
By the definition of $\rho$, we have $|u_0| \ge \rho$.
The perturbation condition implies that
\begin{equation}
|u - u_0| = \left| \bigl|[\mathbf{A}(\Theta)\tilde{\mathbf{S}}]_{m,p}\bigr| - \bigl|[\mathbf{A}(\Theta)\tilde{\mathbf{S}}_0]_{m,p}\bigr| \right| < \rho.
\end{equation}
Since $|u - u_0| < \rho \le |u_0|$, the value $u$ cannot cross zero; hence, $u$ must have the same sign as $u_0$.
Consequently,
\begin{equation}
\operatorname{sign}(u) = \operatorname{sign}(u_0) \implies [\widehat{\mathbf{Y}}_{\mathrm{1bit}}(\tilde{\mathbf{S}})]_{m,p} = [\widehat{\mathbf{Y}}_{\mathrm{1bit}}(\tilde{\mathbf{S}}_0)]_{m,p}.
\end{equation}
Since this holds for all $(m,p)$, we conclude that $\widehat{\mathbf{Y}}_{\mathrm{1bit}}(\tilde{\mathbf{S}})=\widehat{\mathbf{Y}}_{\mathrm{1bit}}(\tilde{\mathbf{S}}_0)$.

To prove the sufficient condition in the variable domain, we utilize the reverse triangle inequality $\bigl| |a|-|b| \bigr| \le |a-b|$ and standard norm inequalities:
\begin{align}
&\bigl\||\mathbf{A}(\Theta)\tilde{\mathbf{S}}|-|\mathbf{A}(\Theta)\tilde{\mathbf{S}}_0|\bigr\|_\infty \nonumber\\
&\quad \le \|\mathbf{A}(\Theta)(\tilde{\mathbf{S}}-\tilde{\mathbf{S}}_0)\|_\infty \nonumber\\
&\quad \le \|\mathbf{A}(\Theta)(\tilde{\mathbf{S}}-\tilde{\mathbf{S}}_0)\|_F \nonumber\\
&\quad \le \|\mathbf{A}(\Theta)\|_2\,\|\tilde{\mathbf{S}}-\tilde{\mathbf{S}}_0\|_F.
\label{eq:norm_chain_app}
\end{align}
Therefore, if $\|\tilde{\mathbf{S}}-\tilde{\mathbf{S}}_0\|_F < \rho / \|\mathbf{A}(\Theta)\|_2$, the right-hand side of \eqref{eq:norm_chain_app} is strictly less than $\rho$, which implies that \eqref{eq:local_delta_condition} holds. This completes the proof.
\end{proof}

\section{Proof of Lemma~\ref{lem:Lipschitz}}
\label{app:proof_Lipschitz}

Recall from \eqref{eq:grad_L_eps} that
\begin{equation}
\label{eq:grad_recall_app}
\nabla L_\epsilon(\tilde{\mathbf{S}})
=\frac{1}{MP}\mathbf{A}(\Theta)^{\mathsf H}
\big(\mathbf{P}_\epsilon(\tilde{\mathbf{S}})\odot \mathbf{W}(\tilde{\mathbf{S}})\big).
\end{equation}
Take any $\tilde{\mathbf{S}}_1,\tilde{\mathbf{S}}_2\in\mathbb{C}^{G\times P}$ and denote
$\mathbf{Z}_i=\mathbf{A}(\Theta)\tilde{\mathbf{S}}_i$ for $i=1,2$.
Then
\begin{align}
&\|\nabla L_\epsilon(\tilde{\mathbf{S}}_1)-\nabla L_\epsilon(\tilde{\mathbf{S}}_2)\|_F \nonumber\\
&\quad =\frac{1}{MP}\Big\|\mathbf{A}(\Theta)^{\mathsf H}
\big(\mathbf{P}_{\epsilon,1}\odot\mathbf{W}_1-\mathbf{P}_{\epsilon,2}\odot\mathbf{W}_2\big)\Big\|_F \nonumber\\
&\quad \le \frac{\|\mathbf{A}(\Theta)\|_2}{MP}\,
\big\|\mathbf{P}_{\epsilon,1}\odot\mathbf{W}_1-\mathbf{P}_{\epsilon,2}\odot\mathbf{W}_2\big\|_F,
\label{eq:lip_step1}
\end{align}
where $\mathbf{P}_{\epsilon,i} = \mathbf{P}_\epsilon(\tilde{\mathbf{S}}_i)$ and
$\mathbf{W}_i = \mathbf{W}(\tilde{\mathbf{S}}_i)$.

Fix an entry $(m,p)$ and write $z_i=[\mathbf{Z}_i]_{m,p}$.
Let $u_i = |z_i|_\epsilon=\sqrt{|z_i|^2+\epsilon^2}$ and
\begin{equation}
\label{eq:pw_def_app}
p_i = \frac{z_i}{u_i},\qquad
w_i = -\beta\,y_{m,p}\,\sigma\!\big(-\beta\,y_{m,p}(u_i-\tau_m)\big),
\end{equation}
so that $[\mathbf{P}_{\epsilon,i}]_{m,p}=p_i$ and $[\mathbf{W}_i]_{m,p}=w_i$.
Then
\begin{equation}
\label{eq:pw_expand_app}
p_1w_1-p_2w_2 = p_1(w_1-w_2) + w_2(p_1-p_2),
\end{equation}
and hence
\begin{equation}
|p_1w_1-p_2w_2|\le |p_1|\,|w_1-w_2| + |w_2|\,|p_1-p_2|.
\label{eq:lip_scalar_split}
\end{equation}

First, $|p_1|\le 1$ since $u_1\ge |z_1|$, and $|w_2|\le \beta$ since $0<\sigma(\cdot)<1$.
Moreover, the sigmoid satisfies $\sigma(t)(1-\sigma(t))\le 1/4$, so the scalar function
$r\mapsto -\beta\,\sigma(-\beta r)$ is $(\beta^2/4)$-Lipschitz. Therefore,
\begin{align}
|w_1-w_2|
&\le \frac{\beta^2}{4}\,|u_1-u_2| \nonumber\\
&\le \frac{\beta^2}{4}\,\bigl||z_1|-|z_2|\bigr| 
\le \frac{\beta^2}{4}\,|z_1-z_2|,
\label{eq:w_lip}
\end{align}
where we used $|u_1-u_2| = \bigl||z_1|_\epsilon-|z_2|_\epsilon\bigr| \le \bigl||z_1|-|z_2|\bigr|$
(since $|\cdot|_\epsilon$ is $1$-Lipschitz in $|\cdot|$)
and the reverse triangle inequality $\bigl||z_1|-|z_2|\bigr| \le |z_1-z_2|$.

Second, the map $g(z) = z/\sqrt{|z|^2+\epsilon^2}$ is $(1/\epsilon)$-Lipschitz, i.e.,
\begin{equation}
|p_1-p_2| = |g(z_1)-g(z_2)| \le \frac{1}{\epsilon}\,|z_1-z_2|.
\label{eq:p_lip}
\end{equation}
Indeed, $g$ is smooth on $\mathbb{R}^2$ and its Jacobian spectral norm is bounded by $1/\epsilon$
because $\sqrt{|z|^2+\epsilon^2}\ge \epsilon$.

Combining \eqref{eq:lip_scalar_split}--\eqref{eq:p_lip} yields
\begin{equation}
|p_1w_1-p_2w_2|\le\left(\frac{\beta^2}{4}+\frac{\beta}{\epsilon}\right)|z_1-z_2|.
\label{eq:scalar_lip_final}
\end{equation}

Applying \eqref{eq:scalar_lip_final} and summing over $(m,p)$ gives
\begin{align}
&\big\|\mathbf{P}_{\epsilon,1}\odot\mathbf{W}_1-\mathbf{P}_{\epsilon,2}\odot\mathbf{W}_2\big\|_F \nonumber\\
&\quad \le\left(\frac{\beta^2}{4}+\frac{\beta}{\epsilon}\right)\|\mathbf{Z}_1-\mathbf{Z}_2\|_F.
\label{eq:lip_step2}
\end{align}
Finally, $\|\mathbf{Z}_1-\mathbf{Z}_2\|_F=\|\mathbf{A}(\Theta)(\tilde{\mathbf{S}}_1-\tilde{\mathbf{S}}_2)\|_F
\le \|\mathbf{A}(\Theta)\|_2\|\tilde{\mathbf{S}}_1-\tilde{\mathbf{S}}_2\|_F$.
Substituting into \eqref{eq:lip_step1}--\eqref{eq:lip_step2} yields
\begin{align}
&\|\nabla L_\epsilon(\tilde{\mathbf{S}}_1)-\nabla L_\epsilon(\tilde{\mathbf{S}}_2)\|_F \nonumber\\
&\quad\le
\frac{\beta\|\mathbf{A}(\Theta)\|_2^2}{MP}\left(\frac{\beta}{4}+\frac{1}{\epsilon}\right)
\|\tilde{\mathbf{S}}_1-\tilde{\mathbf{S}}_2\|_F,
\label{eq:lip_final_app}
\end{align}
which proves the claim. \hfill$\square$

\section{Proof of Lemma~\ref{lem:smoothing_error}}
\label{app:proof_smoothing_error}

Let $\ell(u) = \log(1+e^{-\beta u})$. Since
$\ell'(u)=-\beta\sigma(-\beta u)$ and $0<\sigma(\cdot)<1$, we have
$|\ell'(u)|\le \beta$, i.e., $\ell(\cdot)$ is $\beta$-Lipschitz.

Fix any $(m,p)$ and denote
$a = \big|[\mathbf{A}(\Theta)\tilde{\mathbf{S}}]_{m,p}\big|\ge 0$ and
$a_\epsilon = \sqrt{a^2+\epsilon^2}$.
Then
\begin{equation}
\label{eq:smooth_bound_app}
0\le a_\epsilon-a
=\sqrt{a^2+\epsilon^2}-a
\le \epsilon,
\end{equation}
where the last inequality follows from $\sqrt{a^2+\epsilon^2}\le a+\epsilon$.

Using the $\beta$-Lipschitz property and $|[\mathbf{Y}_{\mathrm{1bit}}]_{m,p}|=1$,
\begin{align}
&\Big|\ell\!\Big([\mathbf{Y}_{\mathrm{1bit}}]_{m,p}(a_\epsilon-\tau_m)\Big)
-\ell\!\Big([\mathbf{Y}_{\mathrm{1bit}}]_{m,p}(a-\tau_m)\Big)\Big| \nonumber\\
&\quad\le \beta\,|a_\epsilon-a|
\le \beta\epsilon.
\label{eq:ell_diff_app}
\end{align}
Averaging over all $(m,p)$ yields
$|L_\epsilon(\tilde{\mathbf{S}})-L(\tilde{\mathbf{S}})|\le \beta\epsilon$.
\hfill$\square$

\section{Proof of Lemma~\ref{lem:descent}}
\label{app:proof_descent}

Since $\nabla L_\epsilon$ is $L_{\mathrm{Lip}}$-Lipschitz continuous, the descent lemma
implies for any $\mathbf{U},\mathbf{V}\in\mathbb{C}^{G\times P}$~\cite{beck2009fista}:
\begin{align}
L_\epsilon(\mathbf{U})
&\le
L_\epsilon(\mathbf{V})
+\operatorname{Re}\langle \nabla L_\epsilon(\mathbf{V}),\mathbf{U}-\mathbf{V}\rangle \nonumber\\
&\quad +\frac{L_{\mathrm{Lip}}}{2}\|\mathbf{U}-\mathbf{V}\|_F^2.
\label{eq:quad_upper_app}
\end{align}

The proximal-gradient update \eqref{eq:pg_update} is equivalently written as
\begin{align}
\tilde{\mathbf{S}}^{(t+1)}
&\in
\arg\min_{\mathbf{U}}\Big\{
\operatorname{Re}\langle \nabla L_\epsilon(\tilde{\mathbf{S}}^{(t)}),\mathbf{U}-\tilde{\mathbf{S}}^{(t)}\rangle \nonumber\\
&\qquad\qquad +\frac{1}{2\mu}\|\mathbf{U}-\tilde{\mathbf{S}}^{(t)}\|_F^2
+\eta\|\mathbf{U}\|_{2,1}
\Big\}.
\label{eq:prox_def_app}
\end{align}
By optimality of $\tilde{\mathbf{S}}^{(t+1)}$, evaluating the minimized objective at
$\mathbf{U}=\tilde{\mathbf{S}}^{(t)}$ yields
\begin{align}
&\operatorname{Re}\langle \nabla L_\epsilon(\tilde{\mathbf{S}}^{(t)}),\tilde{\mathbf{S}}^{(t+1)}-\tilde{\mathbf{S}}^{(t)}\rangle \nonumber\\
&\quad +\frac{1}{2\mu}\|\tilde{\mathbf{S}}^{(t+1)}-\tilde{\mathbf{S}}^{(t)}\|_F^2
+\eta\|\tilde{\mathbf{S}}^{(t+1)}\|_{2,1}
\le
\eta\|\tilde{\mathbf{S}}^{(t)}\|_{2,1}.
\label{eq:prox_ineq_app}
\end{align}

Applying \eqref{eq:quad_upper_app} with $\mathbf{V}=\tilde{\mathbf{S}}^{(t)}$ and
$\mathbf{U}=\tilde{\mathbf{S}}^{(t+1)}$ gives
\begin{align}
L_\epsilon(\tilde{\mathbf{S}}^{(t+1)})
&\le
L_\epsilon(\tilde{\mathbf{S}}^{(t)}) \nonumber\\
&\quad +\operatorname{Re}\langle \nabla L_\epsilon(\tilde{\mathbf{S}}^{(t)}),
\tilde{\mathbf{S}}^{(t+1)}-\tilde{\mathbf{S}}^{(t)}\rangle \nonumber\\
&\quad +\frac{L_{\mathrm{Lip}}}{2}\|\tilde{\mathbf{S}}^{(t+1)}-\tilde{\mathbf{S}}^{(t)}\|_F^2.
\label{eq:descent_step1_app}
\end{align}
Combining \eqref{eq:descent_step1_app} with \eqref{eq:prox_ineq_app} yields
\begin{align}
F_\epsilon(\tilde{\mathbf{S}}^{(t+1)})
&=L_\epsilon(\tilde{\mathbf{S}}^{(t+1)})+\eta\|\tilde{\mathbf{S}}^{(t+1)}\|_{2,1} \nonumber\\
&\le
L_\epsilon(\tilde{\mathbf{S}}^{(t)})+\eta\|\tilde{\mathbf{S}}^{(t)}\|_{2,1} \nonumber\\
&\quad -\left(\frac{1}{2\mu}-\frac{L_{\mathrm{Lip}}}{2}\right)
\|\tilde{\mathbf{S}}^{(t+1)}-\tilde{\mathbf{S}}^{(t)}\|_F^2 \nonumber\\
&=
F_\epsilon(\tilde{\mathbf{S}}^{(t)})
-\frac{1-\mu L_{\mathrm{Lip}}}{2\mu}\,
\|\tilde{\mathbf{S}}^{(t+1)}-\tilde{\mathbf{S}}^{(t)}\|_F^2.
\label{eq:descent_final_app}
\end{align}
If $\mu\in(0,1/L_{\mathrm{Lip}})$, then $1-\mu L_{\mathrm{Lip}}>0$, proving \eqref{eq:descent_ineq}.
\hfill$\square$

\section{Proof of Theorem~\ref{thm:conv_critical}}
\label{app:proof_convergence}

By Lemma~\ref{lem:descent}, $\{F_\epsilon(\tilde{\mathbf{S}}^{(t)})\}$ is monotonically nonincreasing.
Since $F_\epsilon(\tilde{\mathbf{S}})\ge 0$, the limit
$F_\epsilon^\infty = \lim_{t\to\infty}F_\epsilon(\tilde{\mathbf{S}}^{(t)})$ exists.

Summing \eqref{eq:descent_ineq} over $t=0,\ldots,T-1$ yields
\begin{align}
&\sum_{t=0}^{T-1}\|\tilde{\mathbf{S}}^{(t+1)}-\tilde{\mathbf{S}}^{(t)}\|_F^2 \nonumber\\
&\quad \le
\frac{2\mu}{1-\mu L_{\mathrm{Lip}}}
\Big(F_\epsilon(\tilde{\mathbf{S}}^{(0)})-F_\epsilon(\tilde{\mathbf{S}}^{(T)})\Big).
\label{eq:sum_descent_app}
\end{align}
Letting $T\to\infty$ gives
\begin{align}
&\sum_{t=0}^{\infty}\|\tilde{\mathbf{S}}^{(t+1)}-\tilde{\mathbf{S}}^{(t)}\|_F^2 \nonumber\\
&\quad \le
\frac{2\mu}{1-\mu L_{\mathrm{Lip}}}
\Big(F_\epsilon(\tilde{\mathbf{S}}^{(0)})-F_\epsilon^\infty\Big)
<\infty,
\label{eq:sum_infty_app}
\end{align}
which implies $\|\tilde{\mathbf{S}}^{(t+1)}-\tilde{\mathbf{S}}^{(t)}\|_F\to 0$.

All iterations lie in the initial sublevel set
$\mathcal{S}_0=\{\tilde{\mathbf{S}}:F_\epsilon(\tilde{\mathbf{S}})\le F_\epsilon(\tilde{\mathbf{S}}^{(0)})\}$.
As shown in Section~\ref{subsec:convergence}, $\mathcal{S}_0$ is closed and bounded, and hence compact.
Therefore, $\{\tilde{\mathbf{S}}^{(t)}\}$ admits at least one accumulation point:
there exists a subsequence $\{\tilde{\mathbf{S}}^{(t_k)}\}$ such that
$\tilde{\mathbf{S}}^{(t_k)}\to \tilde{\mathbf{S}}^\star$.
Moreover, since $\|\tilde{\mathbf{S}}^{(t+1)}-\tilde{\mathbf{S}}^{(t)}\|_F\to 0$, we also have
$\tilde{\mathbf{S}}^{(t_k+1)}\to \tilde{\mathbf{S}}^\star$.

Next we show that any accumulation point is a critical point of $F_\epsilon$.
By the optimality condition of the proximal step \eqref{eq:pg_update},
\begin{align}
\mathbf{0}&\in
\nabla L_\epsilon(\tilde{\mathbf{S}}^{(t)})+
\frac{1}{\mu}\big(\tilde{\mathbf{S}}^{(t+1)}-\tilde{\mathbf{S}}^{(t)}\big) \nonumber\\
&\quad +\eta\,\partial\|\cdot\|_{2,1}\big(\tilde{\mathbf{S}}^{(t+1)}\big),
\label{eq:opt_cond_pg}
\end{align}
where $\partial\|\cdot\|_{2,1}$ denotes the convex subdifferential.
Equivalently, there exists
$\mathbf{V}^{(t+1)}\in\partial\|\cdot\|_{2,1}(\tilde{\mathbf{S}}^{(t+1)})$ such that
\begin{align}
&\nabla L_\epsilon(\tilde{\mathbf{S}}^{(t)})+
\frac{1}{\mu}\big(\tilde{\mathbf{S}}^{(t+1)}-\tilde{\mathbf{S}}^{(t)}\big) \nonumber\\
&\quad +\eta\,\mathbf{V}^{(t+1)}=\mathbf{0}.
\label{eq:select_subgrad}
\end{align}
Take $t=t_k$ and let $k\to\infty$.
Using the continuity of $\nabla L_\epsilon$,
$\tilde{\mathbf{S}}^{(t_k)}\to \tilde{\mathbf{S}}^\star$,
and $\tilde{\mathbf{S}}^{(t_k+1)}-\tilde{\mathbf{S}}^{(t_k)}\to\mathbf{0}$,
we obtain
\begin{equation}
\nabla L_\epsilon(\tilde{\mathbf{S}}^\star)+\eta\,\mathbf{V}^\star=\mathbf{0}
\quad\text{for some}\quad
\mathbf{V}^\star\in\partial\|\cdot\|_{2,1}(\tilde{\mathbf{S}}^\star),
\end{equation}
where we used the closedness (outer semicontinuity) of the convex subdifferential mapping
$\partial\|\cdot\|_{2,1}$.
Hence,
\begin{align}
\mathbf{0}&\in \nabla L_\epsilon(\tilde{\mathbf{S}}^\star)+\eta\,\partial\|\cdot\|_{2,1}(\tilde{\mathbf{S}}^\star) \nonumber\\
&=\partial F_\epsilon(\tilde{\mathbf{S}}^\star),
\label{eq:critical_point_app}
\end{align}
which shows that every accumulation point is a critical point of $F_\epsilon$.
\hfill$\square$

\bibliographystyle{IEEEtran}
\bibliography{Lu1_complete_fixed}

\end{document}